\documentclass[runningheads, envcountsame]{llncs}
\usepackage[T1]{fontenc}

\usepackage{graphicx}
\usepackage{booktabs} 
\usepackage[ruled]{algorithm2e} 

\SetAlFnt{\small}
\SetAlCapFnt{\small}
\SetAlCapNameFnt{\small}
\SetAlCapHSkip{0pt}
\IncMargin{-\parindent}

\usepackage{tcolorbox}
\usepackage{mathtools}
\usepackage{soul}
\usepackage[hidelinks]{hyperref} 
\usepackage{color}

\usepackage{caption} 
\usepackage{subcaption}
\usepackage{amsmath,amsfonts,amssymb} 
\usepackage{booktabs}
\usepackage{xspace}
\usepackage{tikz}
\usetikzlibrary{tikzmark}
\usepackage{array}
\usepackage[misc,geometry]{ifsym} 

\newcommand\fauxsc[1]{\fauxschelper#1 \relax\relax}
\def\fauxschelper#1 #2\relax{%
  \fauxschelphelp#1\relax\relax%
  \if\relax#2\relax\else\ \fauxschelper#2\relax\fi%
}
\def\Hscale{.85}\def\Vscale{.72}\def\Cscale{1.10}
\def\fauxschelphelp#1#2\relax{%
  \ifnum`#1>``\ifnum`#1<`\{\scalebox{\Hscale}[\Vscale]{\uppercase{#1}}\else%
    \scalebox{\Cscale}[1]{#1}\fi\else\scalebox{\Cscale}[1]{#1}\fi%
  \ifx\relax#2\relax\else\fauxschelphelp#2\relax\fi}

\usepackage{tabularray}
\usepackage{resizegather}

\usepackage{environ} 
\newcommand{\repeattheorem}[1]{%
  \begingroup
  \renewcommand{\thetheorem}{\ref{#1}}%
  \expandafter\expandafter\expandafter\theorem
  \csname reptheorem@#1\endcsname
  \endtheorem
  \endgroup
}

\NewEnviron{reptheorem}[1]{%
  \global\expandafter\xdef\csname reptheorem@#1\endcsname{%
    \unexpanded\expandafter{\BODY}%
  }%
  \expandafter\theorem\BODY\unskip\label{#1}\endtheorem
}

\newcommand{\repeatlemma}[1]{%
  \begingroup
  \renewcommand{\thelemma}{\ref{#1}}%
  \expandafter\expandafter\expandafter\lemma
  \csname replemma@#1\endcsname
  \endlemma
  \endgroup
}

\NewEnviron{replemma}[1]{%
  \global\expandafter\xdef\csname replemma@#1\endcsname{%
    \unexpanded\expandafter{\BODY}%
  }%
  \expandafter\lemma\BODY\unskip\label{#1}\endlemma
}
  
\newcommand{\edge }{\text{ --- }}

\newcommand{\calU}{{\mathcal{U}}}
\newcommand{\calP}{{\mathcal{P}}}
\newcommand{\calC}{{\mathcal{C}}}

\newif\ifwitheverything
\witheverythingfalse

\usepackage{enumitem}
\renewcommand{\labelenumi}{{$\langle${\theenumi}$\rangle$}}

\begin{document}
\title{Recovering Single-Crossing Preferences From Approval Ballots}
%

\author{Andrei Constantinescu\textsuperscript{\thinspace\Letter\thinspace}\orcidID{0009-0005-1708-9376} \and
Roger Wattenhofer \orcidID{0000-0002-6339-3134}}
\authorrunning{A. Constantinescu and R. Wattenhofer}
%
\institute{ETH Zurich, Rämistrasse 101, 8092 Zurich, Switzerland \\
\email{\{aconstantine,wattenhofer\}@ethz.ch}}
\maketitle              
\begin{abstract} 
An electorate with fully-ranked innate preferences casts approval votes over a finite set of alternatives. As a result, only partial information about the true preferences is revealed to the voting authorities. In an effort to understand the nature of the true preferences given only partial information, one might ask whether the unknown innate preferences could possibly be single-crossing.
The existence of a polynomial time algorithm to determine this has been asked as an outstanding problem in the works of Elkind and Lackner \cite{edith_ssc}. We hereby give a polynomial time algorithm determining a single-crossing collection of fully-ranked preferences that could have induced the elicited approval ballots, or reporting the nonexistence thereof. Moreover, we consider the problem of identifying negative instances with a set of forbidden sub-ballots, showing that any such characterization requires infinitely many forbidden configurations.

\end{abstract}
%
%
%

\section{Introduction}\label{sec:intro}

One can express their opinion either in \textit{precise} or in \textit{simple} terms. The scientific world favors precision, but real-world decision, voting and election schemes usually adopt simplicity. Both ways have their disadvantages:~simple voting is often limited in expressivity, requiring voters to make compromises (``I would really like to vote for candidate A, but candidate B has an actual chance of getting elected.''). Precise voting, on the other hand, is too demanding of the voter (``Do I really have to rank all these candidates? I barely know them!''), and hence often not implemented. Indeed, the vast majority of the preference data repository \fauxsc{PrefLib} \cite{mattei_walsh_preflib} consists of partially-elicited preferences.
A modern solution striking a good balance between simplicity and
expressive
power stands in approval voting \cite{approval_book}. Here ballots ask voters to decide whether they approve or disapprove of each candidate; i.e., is the candidate good enough or not? Approval voting has been successfully applied to a number of settings, including among others multi-winner elections and participatory budgeting
(see excellent survey in \cite{lackner_skowron_abcs_book}).

Starting with the seminal work of Arrow \cite{arrow}, a long line of research has been dedicated to showing the impossibility of computing fair social outcomes. However, these impossibility proofs typically hinge on pathological instances that are unlikely to occur in real-world elections, where  
political preferences can, for instance, often be explained using a left-right spectrum. Elections following such a spectrum are known as (one-dimensional) Euclidean. This model has received a number of generalizations enhancing its expressive power, including the single-peaked \cite{black_single_peaked,arrow}
and single-crossing \cite{mirrlees_income_tax,roberts_single_crossing} models, weakening the requirements of the societal axis in two distinct ways. An election is single-crossing if there is an ordering of the voters, called the single-crossing axis, such that, as we sweep through voters in order, preference between any two candidates changes at most once. We omit the definition of single-peaked preferences as it is not required for our work. While the two notions are incomparable, they share many of their attractive properties, including the existence of Condorcet winners
and the polynomial-time computability of winners under election schemes for which this is hard in general, like the rules of Young, Dodgson and Kemeny \cite{handbook_chapters} in the single-winner setting, and those of Chamberlin--Courant and Monroe in the multi-winner one \cite{prz08,lb11}. The first three follow from the existence of Condorcet winners, while the latter two use dynamic programming  \cite{cc_sc_poly_time_edith,constantinescu_elkind_2021,chamberlin_poly_time_sp,chamberlin_poly_time_sp_better,hua_monroe}. When preferences are neither single-peaked nor single-crossing, one can still hope for the removal of a small fraction of voters and/or candidates to allow one of the models to apply.

The models discussed so far are defined for fully-elicited preferences. For approval ballots, their most natural extensions ask whether there are complete-information Euclidean, single-peaked, or single-crossing preferences refining the elicited approval preferences. If so, one calls the elicited preferences \emph{possibly} Euclidean, single-peaked, or single-crossing (short PE, PSP, PSC). Elkind and Lackner \cite{edith_ssc} study these notions, among others, showing that PE and PSP surprisingly coincide, and PSC is more general (in fact the most general considered in their paper). However, one so far unanswered question is whether PSC collections of ballots can be recognized in polynomial time. In contrast, this is known to be possible for the other notions.

A polynomial-time algorithm for recognizing PSC and computing an associated societal axis has a few important implications. First of all, such an axis is required to apply the known polynomial-time algorithms for computing winning committees for the approval-based Chamberlin--Courant and Monroe rules under single-crossing preferences \cite{cc_sc_poly_time_edith,constantinescu_elkind_2021,hua_monroe}. Moreover, Pierczy{\'{n}}ski and Skowron \cite{core_top_monotonic} show that core-stable committees always exist for approval elections under two definitions more restrictive than PSC (one of which being PE = PSP). It is likely that the same will be shown in the future for PSC, as no counterexample is known even for arbitrary approval elections,\footnote{This is perhaps the main open question in the literature on proportional representation \cite{lackner_skowron_abcs_book}.} and an algorithm computing the PSC axis will likely be a necessary component of an algorithm computing core-stable committees for PSC preferences. Finally, such an algorithm can be used by electoral authorities to study the political landscape emerging from an election.

\textbf{Our Contribution}. We give a polynomial time algorithm that takes as input a collection of approval ballots and outputs a collection of fully-ranked ballots with the single-crossing property that could have produced the observed approval votes, reporting accordingly in case this is not possible. In the process, our approach also leads to a novel FPT algorithm for the non-betweenness problem \cite{non_bet_np_complete}. Our problem has been posed as open by Elkind and Lackner \cite{edith_ssc}.\footnote{Recently, we learned that a solution has in fact been proposed as early as 1979 in the context of the simple plant location problem, by Beresnev and Davydov \cite{beresnev}. This paper is only available in Russian, and Russian-speaking experts seem to believe that the paper is likely missing steps in the arguments. Beresnev and Davydov \cite{beresnev} is referenced in \cite{gerhard}, but without details.}

Moreover, in Appendix \ref{sect:impossibility}, we show that, while single-crossing elections admit a characterization in terms of a set of two small forbidden subelections \cite{characterization_minors_single_crossing}, such a characterization ceases to exist for PSC elections. We show this by exhibiting an infinite set of approval elections that are not PSC, but all their proper subelections are.

\textbf{Related Work}. Checking whether an election (with fully-elicited strictly ranked preferences) is single-crossing can be achieved in polynomial time \cite{elkind_reco_sc,characterization_minors_single_crossing}, in fact even in near-linear time \cite{survey_restricted}. Checking whether an election can be embedded on a left-right spectrum (i.e., whether it is $1$-Euclidean) can also be done in polynomial time \cite{poly_time_one_dimensional,single_peaked_reco_2,euclidean_reco_2}. In both cases, finding witnessing axes/embeddings is also polynomial. On the other hand, deciding whether an embedding exists in $d$ dimensions becomes NP-hard as soon as $d \geq 2$ \cite{dominik_2d_is_np_hard}, and heuristic accounts support the difficulty of fitting such higher-dimensional models in practice 
\cite{busing_phd_thesis}. 

Single-peaked elections can also be recognized in polynomial time
\cite{single_peaked_reco_1,single_peaked_reco_2,single_peaked_reco_3,survey_restricted}. A natural question related to our work is how difficult it is to recover 1-Euclidean preferences or single-peaked preferences from approval ballots. The two notions coincide (PE = PSP), and polynomial-time algorithms have been proposed \cite{edith_ssc,possibly_single_peaked_first_pointed_out_how_to_reco} by reduction to the consecutive ones problem (which admits a linear-time algorithm, although using rather complex machinery \cite{booth_lueker}).
In contrast, our algorithm for PSC proceeds directly, using non-betweenness only as a lens. An interesting model that encompasses both single-crossing and single-peaked elections is that of top-monotonic elections \cite{barbera_top_monotonic}. Such elections can be recognized in polynomial time \cite{top_monotonic} (also the paper whose techniques are closest to ours).

For the case of multi-valued approval, where voters can choose between more than two judgments when filling in the ballot, Fitzsimmons \cite{fitzsimmons_weak_order_single_peaked} provides a polynomial time recognition algorithm for possibly single-peaked preferences, working even for the case of an unbounded number of possible judgments; i.e., voters give a score between 1 and the number of candidates to each candidate (can also be thought as voters reporting rankings with ties). Their method proceeds by reducing to the consecutive ones problem. Getting similar results for the single-crossing case might seem desirable, but, as shown in \cite{edith_incomplete_sc}, this is NP-hard for an unbounded number of possible judgments.

Working with approval ballots requires rather different social aggregators to those taking fully-ranked profiles as input. These aggregators have been recently extensively studied in the context of multi-winner elections; see the comprehensive survey of both normative and computational results of Lackner and Skowron \cite{lackner_skowron_abcs_book}. In an effort to translate the known models of voting to approval elections, Elkind and Lackner \cite{edith_ssc} consider twelve models catered to approval voting, the most general being the domain of possibly single-crossing elections. For all others, their paper shows polynomial-time recognition,
while Terzopoulou et al.~\cite{forbidden_for_approval} show that the least restrictive of the others admit finite forbidden substructures characterizations. We resolve the open case of PSC and show that it admits no finite forbidden substructures characterization.

Finally, the problem of recognizing possibly single-crossing approval elections has a similar flavor to the recent work of Mandal et al.~\cite{distortion_of_ballot} on distortion induced by the preference elicitation mechanism, thus making it interesting and natural to study. In particular, instead of being interested in recovering voter utilities, here the objective is recovering the single-crossing axis of the innate preferences, mitigating the information lost when the electorate casts their votes through approval ballots; i.e., recovering utilities inducing a single-crossing profile.

\begin{figure}[t]
    \centering
    \begin{subfigure}{.25\textwidth}
        \centering
        \begin{align*}
            \calU     = \{& 1, 2, 3, 4, 5, 6\} \\
            \calC = \{& (1, 2, 3), (2, 3, 4), \\ 
                      & (2, 4, 3), (4, 2, 5), \\
                      & (1, 5, 6)\}
        \end{align*}
        \vspace{0.09cm}
        \caption{Input instance.}
        \label{fig:example-for-intro-nb-constraints}
    \end{subfigure}
    \begin{subfigure}{.45\textwidth}
        \centering
        \includegraphics[width=.95\linewidth]{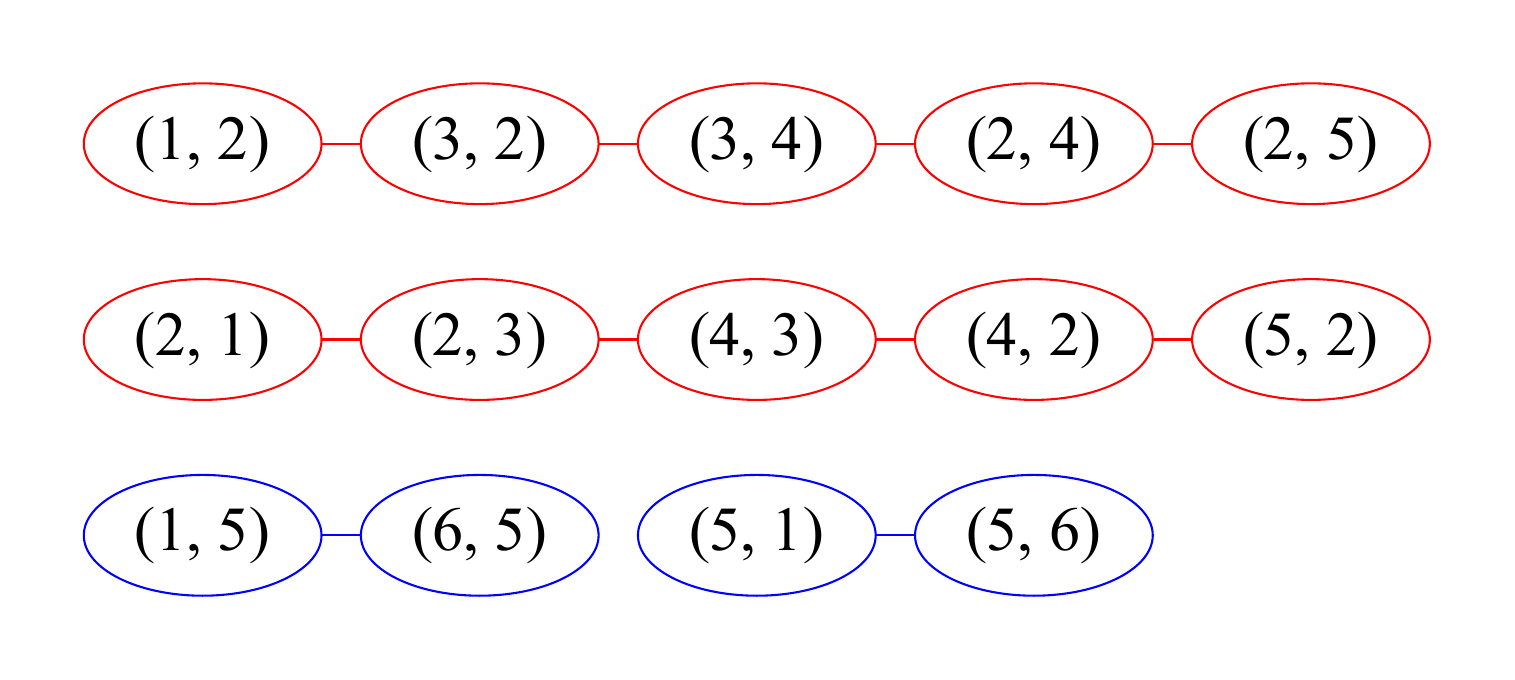}
        \vspace{0.45cm}
        \caption{Formula graph, excl.~singletons.}
        \label{fig:example-for-intro-formula-graph}
    \end{subfigure}
    \begin{subfigure}{.25\textwidth}
        \centering
        \includegraphics[width=.95\linewidth]{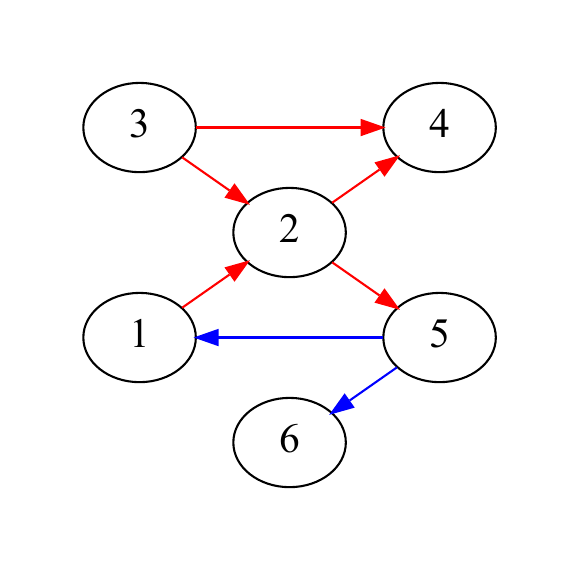}
        \caption{Colorful graph.}
        \label{fig:example-for-intro-colorful-graph}
    \end{subfigure}
    \caption{Example of the construction of the formula and colorful graphs. Consider the instance of the non-betweenness problem given by $\calU$ and $\calC$ in Fig.~\ref{fig:example-for-intro-nb-constraints}. The corresponding formula graph in Fig.~\ref{fig:example-for-intro-formula-graph} is constructed by adding edges $(a, b) \edge (c, b)$ and $(b, a) \edge (b, c)$ for each triple $(a, b, c) \in \calC.$ Isolated nodes have been omitted. Note that that no connected component of the formula graph contains both vertices $(a, b)$ and $(b, a)$ for any $a, b \in \calU.$ Subsequently, the connected components split into two pairs of complementary components:~the two top long red chains and the two bottom short blue chains. The corresponding colorful graph in Fig.~\ref{fig:example-for-intro-colorful-graph} is constructed by (arbitrarily) taking the top red component and the right blue component and introducing directed edges corresponding to the ordered pairs in the component. To solve the 
    instance in Fig.~\ref{fig:example-for-intro-nb-constraints} one needs to decide for each color whether or not to flip the direction of all edges of that color such that the resulting graph is acyclic. In our case, to avoid the cycle $1 \rightarrow 2 \rightarrow 5 \rightarrow 1,$ the possible solutions are to either flip all blue edges or all red edges (but not both, as that leads to cycle $1 \rightarrow 5 \rightarrow 2 \rightarrow 1$ instead). 
    }
    \label{fig:example-for-intro}
\end{figure}

\textbf{Technical Overview}. Our approach for proving the main result proceeds by first translating the input into an instance of the non-betweenness problem, which asks given a ground set $\calU$ and a set $\calC$ of triples $(a, b, c)$ over $\calU$ to compute a linear (i.e., total, irreflexive, antisymmetric, transitive) ordering of the elements of $\calU$ such that for none of the triples does $b$ come between $a$ and $c.$ While non-betweenness is NP-hard in general \cite{non_bet_np_complete}, for the type of instances produced by our reduction the problem will turn out to be polynomial. To show this, we begin with an approach to non-betweenness based on boolean satisfiability. In particular, we introduce for every pair of distinct entries $a, b \in \calU$ a boolean variable $x_{a, b}$ standing for whether $a$ comes before $b$ in the ordering. Antisymmetry can be captured by adding clauses requiring that $x_{a, b} = \overline{x_{b, a}},$ while for every triple $(a, b, c) \in \calC$ one can show that non-betweenness is enforced by adding clauses ensuring that $x_{a, b} = x_{c, b}.$ However, transitivity is substantially trickier to model without resorting to 3-literal clauses. To overcome this difficulty, taking a step back, we notice that except for transitivity all constraints are equalities. We take advantage of this fact by constructing the so-called \emph{formula graph}, which has a vertex $(a, b)$ for each variable $x_{a, b}$ in the boolean formula. Edges in the formula graph correspond to the equalities $x_{a, b} = x_{c, b}$ in the boolean formulation above, together with the inferred opposite equalities $x_{b, a} = x_{b, c}$. If some connected component contains both $x_{a, b}$ and $x_{b, a},$ then the instance is not satisfiable. Otherwise, we introduce yet another construct, the \emph{colorful graph}, which is an edge-colored directed graph with vertex set $G.$ One can show that at this point connected components in the formula graph come in complementary pairs $\{S, \overline{S}\},$ where $(a, b) \in S$ whenever $(b, a) \in \overline{S}.$ For each such pair we create a new color, the edges of this color corresponding to the ordered pairs in set $S.$ Some colors will only have one edge, which can be safely removed. Our two graph concepts are illustrated in Fig.~\ref{fig:example-for-intro}. We then show that satisfying assignments of the boolean formula, this time including transitivity, correspond to acyclic orientations of the colorful graph where edges of each color are either all reversed, or all left in the original orientation. Deciding if such orientations of the colorful graph exist is still NP-hard, but exhausting over all possible orientations already shows that the non-betweenness problem is FPT with respect to the number of colors, a fact which might be of independent interest. At this point, we start studying the structure of the formula and colorful graphs induced by approval elections, showing the surprising result that, unless there is a monochromatic cycle, an acyclic orientation exists and can be computed in polynomial time, from which the election is possibly single-crossing. Proving this result constitutes the most demanding part of our work, and relies on a mixture of combinatorial and computer-aided techniques. Most notably, along the way we prove that if some orientation of the colorful graph is not acyclic, then the induced cycle is either monochromatic or of length three and not repeating colors. With further work, we then prove a strong structural result about colors taking part in three-colored triangles, allowing them to be consistently oriented to complete the proof.

\section{Preliminaries}\label{sect:prelims}

For a non-negative integer $n$, we write $[n]$ for the set $\{1, 2, \ldots, n\}$.
Given a statement $S$, we write $[S]$ for the Iverson bracket:~$[S] = 1$ if $S$ holds, and 0 otherwise.

Given a finite set $A,$ a \emph{partial} order $\succ$ over $A$ is an irreflexive, antisymmetric, transitive relation over $A$. Two values $a, b \in A$ are \emph{comparable} under $\succ$ if either $a \succ b$ or $b \succ a$; otherwise, they are \emph{incomparable}, in which case we write $a \approx_\succ b.$ When any two distinct elements are comparable (i.e., the order is total), order $\succ$ is called a \emph{linear order}. A partial order $\succ$ is called a \emph{weak order} if $\approx_\succ$ forms an equivalence relation (i.e., it is transitive). One can think of weak orders as linear orders allowing for ties, corresponding to incomparable elements. For this reason, when $\succ$ is a weak order,
whenever $a \approx_\succ b$ we say that $\succ$ is \emph{indifferent} between $a$ and $b.$ We call \emph{indifference classes} the equivalence classes induced by $a \approx_\succ b$. We call a weak order with at most $k$ indifference classes a \emph{$k$-weak order}. Note that an order $\succ$ is $k$-weak if and only if $\succ$ can be induced by a \emph{scoring} function $s : A \to [k]$; i.e., $a \succ b$ if and only if $s(a) > s(b).$ Hence, $2$-weak orders are precisely the approval ballots. Given partial orders $\succ$ and $\succ'$, we say that $\succ'$ \emph{extends} $\succ$ if $a \succ' b$ whenever $a \succ b$.

In our setting, an assembly of voters $V = [n]$ expresses their preferences over a set of candidates (alternatives) $C = [m].$ The expressed preferences of voter $i \in V$ consist of a weak order $\succ_i$ over the set of candidates $C.$ The collection of all voter preferences $\calP = (\succ_i)_{i \in [n]}$ is known as the \emph{preference profile}. When $\calP$ consists solely of linear orders, we say that the electorate has \emph{fully-ranked} preferences. When $\calP$ consists solely of $2$-weak orders, we say the electorate has \emph{approval} preferences, in which case we will often see $\calP$ as an $m \times n$ \emph{approval matrix},
where $\calP[c, v] = 1$ if voter $v$ approves of candidate $c$, and $0$ otherwise; i.e., column $v$ contains the approval ballot cast by voter $v$. Given two profiles $\calP$ and $\calP',$ we say that $\calP'$ is a \emph{subprofile} of $\calP$ if $\calP'$ can be obtained from $\calP$ by removing voters and/or candidates. For approval ballots, subprofiles correspond to removing rows and columns from $\calP.$ If $\calP' \neq \calP$ in the previous, then $\calP$ is a \emph{proper} subprofile of $\calP$.

A profile of weak orders $\calP$ is said to be \emph{seemingly single-crossing} (SSC) with respect to a linear order $\triangleleft$ over the set of voters $V,$ called the \emph{axis}, if there are no three voters $i, j, k$ such that $i \triangleleft j \triangleleft k$ and candidates $a, b$ such that $a \succ_i b, b \succ_j a$ and $a \succ_k b$. If additionally profile $\calP$ consists solely of linear orders, then under the same conditions $\calP$ is \emph{single-crossing} (SC) with respect to $\triangleleft$. A profile $\calP$ of weak orders is said to be \emph{possibly single-crossing} (PSC) with respect to a linear order $\triangleleft$ of the voters if for all voters $i \in V$ the preference order $\succ_i$ can be extended to a linear order $\succ_i'$ such that $(\succ_i')_{i \in [n]}$ is single-crossing with respect to $\triangleleft.$ A profile $\calP$ is (seemingly/possibly) single-crossing if it is (seemingly/possibly) single-crossing with respect to some linear order $\triangleleft$ over $V$. For linear orders, the three notions coincide. More surprisingly, Elkind et al.~\cite{edith_incomplete_sc} show that for weak orders the notions of SSC and PSC coincide.\footnote{For general partial orders, they show that PSC implies SSC, but not conversely.}
Moreover, they show that given a linear order $\triangleleft$ witnessing SSC for profile $\calP$, one can compute in polynomial time for each voter $i \in V$ a linear extension $\succ_i'$ of $\succ_i$ such that $(\succ_i')_{i \in [n]}$ is single-crossing with respect to $\triangleleft.$ In other words, knowing a seemingly single-crossing axis is enough to compute
fully-ranked ballots that could have produced the cast ballots in polynomial time. For completeness, we give a streamlined version of their argument in the Appendix \ref{app:equiv-ssc-psc}. Therefore, our focus will be on computing, given a profile $\calP$ of weak orders, an SSC axis, or deciding the nonexistence thereof. For general weak orders, Elkind et al.~\cite{edith_incomplete_sc} show that deciding axis existence is NP-hard, but they leave it open for $k$-weak orders for $k \geq 2$, the case $k = 2$ corresponding with the approval voting case.

Given a finite ground set $\calU$, a \emph{non-betweenness} (NB) constraint over $\calU$ is a triple of distinct elements $(i, j, k) \in \calU^3.$ A linear order $\triangleleft$ over $\calU$ satisfies a collection $\calC$ of NB constraints if for every $(i, j, k) \in \calC$ it does not hold that $i \triangleleft j \triangleleft k$ or $k \triangleleft j \triangleleft i;$ i.e., $j$ is not between $i$ and $k$ in $\triangleleft$. The problem of deciding whether a linear order $\triangleleft$ satisfying $\calC$ exists is NP-hard \cite{non_bet_np_complete}.\footnote{ This is used to prove the hardness of our problem for general weak orders, but the argument requires an unbounded number of indifference classes, so it does not work for bounded $k.$} Checking for the SSC property reduces to verifying whether a set of NB constraints is satisfiable:~given a profile $\calP$, construct a set of NB constraints $\calC_\calP$ over $V$ such that $(i, j, k) \in \calC_\calP$ iff there are two alternatives $a, b \in C$ such that $a \succ_i b, b \succ_j a$ and $a \succ_k b$. Then, we have the following, which is intuitively enabled by the fact that SSC is invariant with respect to reversing the axis:

\begin{replemma}{lemma:reduction_to_nb} $\calP$ is SSC with respect to a linear order $\triangleleft$ 
iff $\calC_\calP$ is satisfied by $\triangleleft.$
\end{replemma}

\section{A Boolean Encoding of Non-Betweenness Constraints}\label{sect:bc-to-boolean-sat}

In this section, we develop a general technique for deciding the satisfiability of an arbitrary set $\calC \subseteq V^3$ of NB constraints over $V.$ Naturally, our approach does not lead to a polynomial time algorithm for checking satisfiability, as the problem is NP-hard, but analysis in later sections will show that when applied to NB constraint sets $\calC_\calP$ originating from approval ballots, it leads to a polynomial algorithm for deciding satisfiability. 

First, note that a NB constraint $(i, j, k)$ is satisfied by a linear order $\triangleleft$ if and only if $[i \triangleleft j] = [k \triangleleft j].$ Moreover, order $\triangleleft$ is uniquely determined by the values $[i \triangleleft j]$ for $i \neq j.$ These facts suggest a reformulation of the problem of satisfying NB constraints in terms of the values $[i \triangleleft j],$ which we do in the following. Given $\calC,$ we construct a boolean formula $\Phi_\calC$ with a variable $x_{i, j}$ for every ordered pair of distinct elements $(i, j) \in V^2.$ The clauses of $\Phi_\calC$ comprise of the following three constraint sets, corresponding to antisymmetry, agreement with $\calC$ and transitivity, respectively:

\begin{enumerate}[ref=\labelenumi]
    \item\label{bool_form_1} For every pair of distinct elements $(i, j) \in V^2$ add the constraint $x_{i, j} = \overline{x_{j, i}}.$ 
    \item\label{bool_form_2} For every triple $(i, j, k) \in \calC$, add the constraint $x_{i, j} = x_{k, j}$.
    \item\label{bool_form_3} For every triple of distinct $(i, j, k) \in V^3$ add the constraint ${(x_{i, j} \land x_{j, k}) \rightarrow x_{i, k}}.$
\end{enumerate}

In the above we used $=$ for the bi-implication operator $\leftrightarrow,$ and we wrote $\overline{x}$ for logical negation, also commonly denoted by $\lnot x$. Note that, under the presence of set \ref{bool_form_1}, set \ref{bool_form_2} can be reformulated as below since 
$x_{j, i} = \overline{x_{i, j}} = \overline{x_{k, j}} = x_{j, k}$:
\begin{enumerate}[label=$\langle2'\rangle$, ref=\labelenumi]
    \item\label{bool_form_2_reformulated} For every triple $(i, j, k) \in \calC$, add the constraints $x_{i, j} = x_{k, j}$ and $x_{j, i} = x_{j, k}.$
\end{enumerate}

Henceforth, we will mostly use \ref{bool_form_2_reformulated}, but \ref{bool_form_2} will still be more convenient in a few cases. The following lemma formally establishes the correspondence between linear orders satisfying $\calC$ and satisfying assignments of $\Phi_\calC$. The proof is essentially straightforward by following the definitions, so we leave it for the appendix.

\begin{replemma}{lemma_boolean_formula_encodes_ordering} Linear orders $\triangleleft$ satisfying the NB constraints in $\calC$ correspond bijectively to satisfying assignments of $\Phi_\calC$ by $x_{i, j} = [i \triangleleft j].$ Hence, $\calC$ is satisfiable
iff $\Phi_\calC$ is satisfiable. 
\end{replemma}

Is this the end of the story? Not quite, because constraint set \ref{bool_form_3} encodes transitivity as $(x_{i, j} \land x_{j, k}) \rightarrow x_{i, k} \equiv (\overline{x_{i, j}} \lor \overline{x_{j, k}} \lor x_{i, k}),$ which is a three-literal clause. However, constraint sets \ref{bool_form_1} and \ref{bool_form_2} only require two-literal clauses, so there is hope. The approach presented so far is similar in spirit to the approach of Magiera and Faliszewski \cite{top_monotonic} for recognizing top-monotonic elections. The main difficulty in their case also stems from enforcing transitivity, which similarly seems to require 3-literal clauses at a first glance. The essential observation that they make to progress is that whenever $(i, j, k) \in \calC$, then the six entries in \ref{bool_form_3} enforcing transitivity for the unordered triple $\{i, j, k\}$ are no longer required, a fact which also holds true in our case. If for every unordered triple we would have at least one NB constraint involving its elements, then none of the transitivity constraints would be required, and the problem could simply be solved by any polynomial 2-CNF solver. This is the main idea of the approach taken in \cite{top_monotonic}. In our case, on the other hand, it turns out that it is surprisingly tricky to completely eliminate the transitivity constraints, so we will need more insight to deal with them gracefully.

\section{Resolving Constraint Sets \ref{bool_form_1} and \ref{bool_form_2} Using Connected Components, The Formula Graph} \label{sect:formula-graph}

Before tackling transitivity, we first turn our attention to constraint sets \ref{bool_form_1} and \ref{bool_form_2} in isolation, providing a precise characterization of the satisfying assignments of $\Phi_\calC$ in the absence of constraint set \ref{bool_form_3}. This characterization will be instrumental later on when reasoning about the transitivity constraints. In contrast, here Magiera and Faliszewski \cite{top_monotonic} directly employ a 2-CNF solver in order to tell whether sets \ref{bool_form_1} and \ref{bool_form_2} are satisfiable in isolation, but doing so only produces one satisfying assignment, rather than a compact representation of all of them, as will be the case for us. Our approach will hinge on the observation that constraint set \ref{bool_form_2} only consists of equalities. In what follows, we will work with constraint set \ref{bool_form_2_reformulated} instead of \ref{bool_form_2}.

To begin, we define the \emph{formula graph} of $\calC$, denoted by $G_\calC.$ The vertex set $V(G_\calC)$ of $G_\calC$ consists of all ordered pairs $(i, j) \in V^2$ with $i \neq j$; i.e., one vertex per variable in the formula. For a vertex $(i, j)$ in $G_\calC$ we define the \emph{complementary} vertex $\overline{(i, j)}$ to be $(j, i).$ Note that this is a syntactic notation, and is not to be confused with $\overline{x_{i, j}}$, which signifies the logical negation of a variable.
For every constraint $x_u = x_v$ in \ref{bool_form_2_reformulated}, where $u$ and $v$ are pairs of distinct elements in $V$, we add an undirected edge $u \edge v$ to $G_\calC$. Intuitively, each edge in $G_\calC$ signifies an equality that needs to hold in the satisfying assignments of \ref{bool_form_1} and \ref{bool_form_2_reformulated}. We begin with two preliminary lemmas concerning the connected components of the formula graph.

\begin{replemma}{lemma_connected} An assignment $(x_u)_{u \in V(G_\calC)}$ satisfies constraint set \ref{bool_form_2_reformulated} if and only if for each connected component $S$ of $G_\calC$ it holds that $x_u = x_v$ for all $u, v \in S$.
\end{replemma}
\begin{replemma}{lemma_complement_is_cc} Let $S$ be a connected component in $G_\calC$, then $\overline{S} := \{\overline{s} : s \in S \}$ is also a connected component in $G_\calC$. Note moreover that $S = \overline{S}$ or $S \cap \overline{S} = \varnothing.$
\end{replemma}

Hence, in light of Lemma \ref{lemma_connected}, if for any connected component $S$ it holds that $S = \overline{S}$, then \ref{bool_form_1} and \ref{bool_form_2_reformulated} cannot be simultaneously satisfied. Otherwise, the connected components of $G_\calC$ can be paired up into ``complementary'' pairs of distinct connected components $\{S, \overline{S}\}$. In particular, $V(G_\calC) = \cup_{i = 1}^k\left(S_i \cup \overline{S_i}\right)$ is a \emph{complementary pairs partition} of the graph, where $(S_i)_{i \in [k]}$ are connected components of $G_\calC$, one per complementary pair. Note that this partition is unique up to changing the roles of $S_i$ and $\overline{S_i}$ for any $i,$ so we will for simplicity refer to it as ``the'' complementary pairs partition. With this in mind, one can notice that solutions to sets \ref{bool_form_1} and \ref{bool_form_2_reformulated} correspond to choosing for each complementary pair $\{S, \overline{S}\}$ one of the two components and setting to true variables corresponding to it and to false variables corresponding to the other component. This is formalized in the following proposition.

\begin{proposition}\label{thm:full_charact}
If $S = \overline{S}$ for any connected component of $G_\calC$, then there are no assignments satisfying \ref{bool_form_1} and \ref{bool_form_2_reformulated}. Otherwise, let $V(G_\calC) = \cup_{i = 1}^k\left(S_i \cup \overline{S_i}\right)$ be the complementary pairs partition. Then, an assignment satisfies \ref{bool_form_1} and \ref{bool_form_2_reformulated} iff it is obtained from a tuple $(s_1, \ldots, s_k) \in \{0, 1\}^k$ by setting all variables $x_u$ to $s_i$ for $u \in S_i$ and to $\overline{s_i}$ for $u \in \overline{S_i}.$
\end{proposition}

Note that the condition $S \neq \overline{S}$ for all connected components $S$ is equivalent to checking that for all $u \in V(G_\calC)$ vertices $u$ and $\overline{u}$ are in different connected components.
Hence, checking whether solutions to \ref{bool_form_1} and \ref{bool_form_2_reformulated} exist is equivalent to checking for the latter, which the reader might recognize as the condition used for checking the satisfiability of 
general 2-CNFs.\footnote{If we replace our undirected graph by the directed implications graph and the word ``connected'' by ``strongly-connected.''} However, for general 2-CNFs no characterization of the solution space is known since even counting solutions is \#P-complete (while in our case we know that there are exactly $2^k$ solutions).

Henceforth, we will assume that $V(G_\calC) = \cup_{i = 1}^k\left(S_i \cup \overline{S_i}\right)$ is the complementary pairs partition (otherwise, just report that constraints in $\calC$ cannot be satisfied). Then, Proposition \ref{thm:full_charact} gives that the assignments satisfying \ref{bool_form_1} and \ref{bool_form_2_reformulated} are precisely those for which a single value $x_{S_i}$ is assigned to each connected component $S_i$ and, similarly, for the complementary component $x_{\overline{S_i}} = \overline{x_{S_i}}$ holds. Note that this already shows that checking the satisfiability of a set of NB constraints is FPT with respect to the number $2k$ of connected components, since one could just try out all $2^k$ satisfying assignments for \ref{bool_form_1} and \ref{bool_form_2_reformulated} and check for transitivity in each case.

\section{Transitivity by Acyclic Orientations, the Colorful Graph}

In this section, we provide an equivalent view of the satisfying assignments of sets \ref{bool_form_1} and \ref{bool_form_2_reformulated} as edge-orientations of a certain edge-colored graph. In particular, all edges have a base orientation, and the orientations satisfying \ref{bool_form_1} and \ref{bool_form_2_reformulated} correspond to selecting for each color whether to leave all edges as they are, or reverse the direction of all edges of that color. Edge orientations additionally satisfying constraint set \ref{bool_form_3} will be those inducing a directed acyclic graph.

To begin, recall that $V(G_\calC) = \cup_{i = 1}^k\left(S_i \cup \overline{S_i}\right)$ is the complementary pairs partition. Consider some connected component $S_i.$ This component consists of a set of ordered pairs $S_i = \{(i_1, j_1), (i_2, j_2), \ldots, (i_\ell, j_\ell)\}$. To have a satisfying assignment of \ref{bool_form_1} and \ref{bool_form_2_reformulated}, it has to be the case that either $x_{i_1, j_1} = \ldots = x_{i_\ell, j_\ell} = 1,$ or $x_{i_1, j_1} = \ldots = x_{i_\ell, j_\ell} = 0.$ Both choices lead to satisfying assignments, independently of the choices made for the other components.
By Lemma \ref{lemma_boolean_formula_encodes_ordering}, the previous gives that 
either $[i_1 \triangleleft j_1] = \ldots = [i_\ell \triangleleft j_\ell] = 1,$ or $[i_1 \triangleleft j_1] = \ldots = [i_\ell \triangleleft j_\ell] = 0,$ and the choice is independent across $S_i$'s.

Any linear order $\triangleleft$ over $V$ can be seen as a directed tournament graph with vertex set $V$, where we draw a directed edge $i \rightarrow j$ if $i \triangleleft j$ and a directed edge $j \rightarrow i$ if $j \triangleleft i.$ Consequently, the constraint $[i_1 \triangleleft j_1] = \ldots = [i_\ell \triangleleft j_\ell]$ induced by some connected component $S_i$ of the formula graph is equivalent in the tournament of $\triangleleft$ to the fact that we either have edges $i_1 \rightarrow j_1, \ldots, i_\ell \rightarrow j_\ell,$ or edges $j_1 \rightarrow i_1, \ldots, j_\ell \rightarrow i_\ell.$ In other words, every set $S_i$ is a directed set of edges that in the tournament of $\triangleleft$ either has to have the given orientation, or exactly the reverse orientation. We call these two options for $S_i$ its two possible \emph{orientations}. We call a choice between the two options for each of the sets $S_1, \ldots, S_k$ an \emph{orientation} of all edge sets.
The following shows that out of those orientations, the ones that
additionally
satisfy set \ref{bool_form_3} are those where the resulting directed tournament graph is acyclic. For brevity, we call such orientations \emph{acyclic}.

\begin{proposition} \label{prop_feasible_iff_acyclic} A linear order $\triangleleft$ over $V$ satisfies the NB constraints in $\calC$ if and only if it is given by an acyclic orientation of the edge sets $S_1, \ldots, S_k.$
\end{proposition}

Some edge sets (connected components in the formula graph) $S_i$ will be singletons; i.e., $|S_i| = 1$. Such sets are not relevant for deciding whether an acyclic orientation exists, because singletons can just be oriented freely after everything else has been oriented. More formally, without loss of generality, assume that edge sets $S_1, \ldots, S_\ell$ are non-singletons; i.e., $|S_i| > 1$ for $i \in [\ell]$; and that edge sets $S_{\ell + 1}, \ldots, S_k$ are singletons; i.e., $|S_i| = 1$ for $\ell < i \leq k$. Then, we have the following:

\begin{lemma} \label{lemma_acyclic_and_non_singleton} An acyclic orientation of edge sets $S_1, \ldots, S_k$ exists if and only if an acyclic orientation of edge sets $S_1, \ldots, S_\ell$ exists.
\end{lemma}

Armed as such, we now introduce the \emph{colorful graph} corresponding to the NB constraint set $\calC.$ The colorful graph is a directed graph with vertex set $V$ and edges colored in $\ell$ colors, identified by the numbers in $[\ell].$ For each $i \in [\ell]$ the edges colored with color $i$ are those in edge set $S_i.$

Note how for any $u, v \in V$ with $u \neq v$ at most one of the edges $u \rightarrow v$ and $v \rightarrow u$ can appear in $S_1, \ldots, S_\ell$. Moreover, an edge cannot occur in more than one of $S_1, \ldots, S_\ell$. As a result, the colorful graph has neither parallel nor anti-parallel edges, and no self-loops. When either edge $a \rightarrow b$ or $b \rightarrow a$ exists in the colorful graph, we say that $a \edge b$ exists in the colorful graph. In this case, we write $c(a, b) = c(b, a)$ to denote the color of the respective oriented edge. 
When the ``undirected'' edge $a \edge b$ exists in the graph, and the corresponding directed edge is $a \rightarrow b$, we say that $a \edge b$ is \emph{oriented} from $a$ to $b.$ Finally, when edge $a \rightarrow b$ ($a\edge b$) bears color $c$, we write $a \xrightarrow[]{c} b$ ($a \overset{c}{\edge } b$). Note that the colorful graph is only defined for NB constraint sets $\calC$ for which the complementary pairs partition 
exists.

We can now reformulate our task in terms of the colorful graph. In particular, we want to check whether the edges of the colorful graph can be reoriented such that the resulting graph is acyclic, under the constraint that for each color we choose between flipping (reversing) the direction of all edges of that color or doing nothing. When flipping the direction of edges of a certain color, we say that we ``flipped'' that color.

\newcommand{\matrixcolumnwidth}{5pt}
\newcommand{\hcr}[1]{{\textcolor{red}{#1}}}
\newcommand{\hcb}[1]{{\textcolor{blue}{#1}}}
\newcommand{\hco}[1]{{\textcolor{orange}{#1}}}
\newcommand{\hc}[1]{{\hco{#1}}}
\begin{figure}[t]
\centering
\begin{subfigure}{.2\textwidth}
    \centering
    \begin{tabular}{*{5}{p{\matrixcolumnwidth}}}
        \textbf{1} & \textbf{2} & \textbf{3} & \textbf{4} & \textbf{5} \\
        \hline
        0 & \textbf{1} & 0 & \textbf{1} & 0 \\
        0 & 0 & 0 & 0 & \textbf{1} \\
        0 & 0 & \textbf{1} & \textbf{1} & 0 \\
        \textbf{1} & \textbf{1} & 0 & \textbf{1} & 0 \\ 
        \textbf{1} & 0 & 0 & 0 & 0 
    \end{tabular}
    \vspace{0.45cm}
    \caption{Profile $\calP_1.$}
    \label{fig:example-colorful-graph-for-approval-a}
\end{subfigure}
\begin{subfigure}{.24\textwidth}
    \centering
    \includegraphics[width=.85\linewidth]{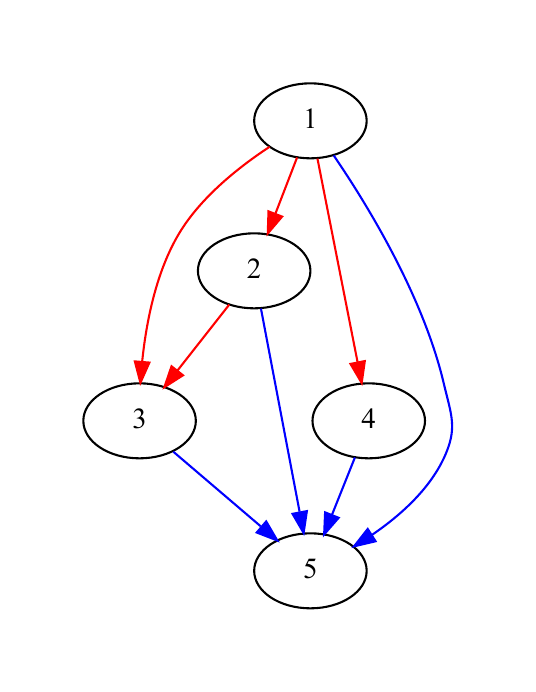}
    \caption{Col.~graph of $\calP_1.$}
    \label{fig:example-colorful-graph-for-approval-b}
\end{subfigure}
\begin{subfigure}{.25\textwidth}
    \centering
    \begin{tabular}{*{7}{p{\matrixcolumnwidth}}}
        \textbf{1} & \textbf{2} & \textbf{3} & \textbf{4} & \textbf{5} & \textbf{6} & \textbf{7} \\
        \hline
        0 & 0 & 0 & 0 & 0 & \textbf{1} & 0 \\
        \textbf{1} & 0 & 0 & \textbf{1} & 0 & 0 & 0 \\
        0 & \textbf{1} & 0 & 0 & 0 & 0 & 0 \\
        0 & 0 & \textbf{1} & 0 & \textbf{1} & 0 & 0 \\
        \textbf{1} & 0 & 0 & 0 & 0 & 0 & 0 \\
        0 & 0 & \textbf{1} & 0 & 0 & \textbf{1} & 0 \\
        \textbf{1} & 0 & 0 & 0 & 0 & 0 & \textbf{1}
    \end{tabular}
    \caption{Profile $\calP_2.$}
    \label{fig:example-colorful-graph-for-approval-c}
\end{subfigure}
\begin{subfigure}{.28\textwidth}
    \centering
    \includegraphics[width=.96\linewidth]{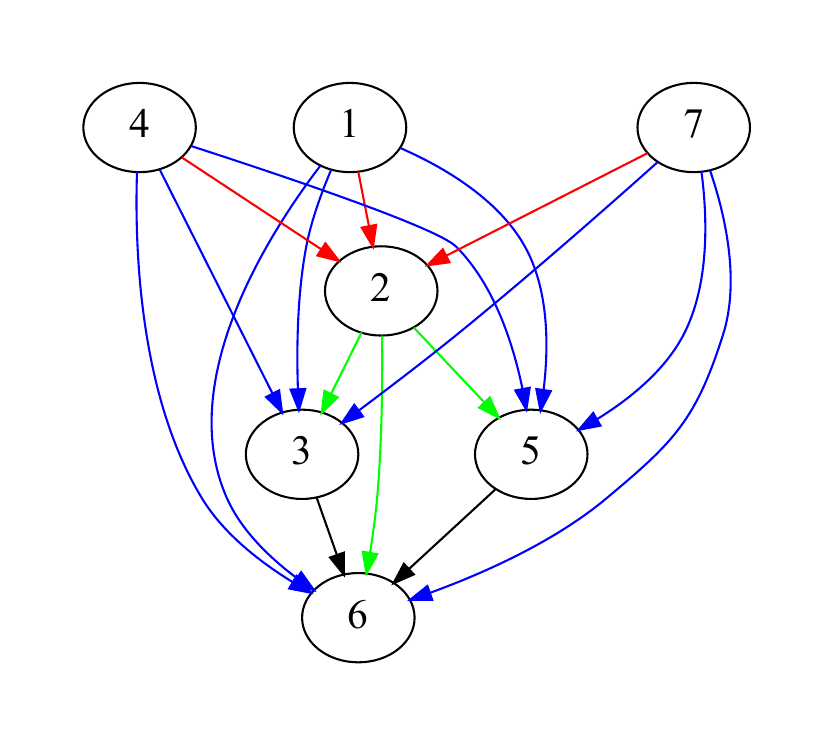}
    \caption{Colorful graph of $\calP_2.$}
    \label{fig:example-colorful-graph-for-approval-d}
\end{subfigure}
\caption{Two example preference profiles $\calP_1, \calP_2,$ together with the colorful graphs they induce. In Fig.~\ref{fig:example-colorful-graph-for-approval-d} all colors are biclique colors. Namely, the red edges are given by $\{1, 4, 7\} \times \{2\},$ the green edges by $\{2\} \times \{3, 5, 6\},$ the black edges by $\{3, 5\} \times \{6\}$ and the blue edges by $\{1, 4, 7\} \times \{3, 5, 6\}.$ Out of these, red, green and black are star biclique. In general, not all colors will be biclique, e.g., in Fig.~\ref{fig:example-colorful-graph-for-approval-b}, only the blue color is biclique.}
\label{fig:example-colorful-graph-for-approval}
\end{figure}

\section{The Case of Approval Ballots}

We now switch our attention to colorful graphs induced by approval preference profiles $\calP;$ i.e., corresponding to the set of NB constraints $\calC_\calP.$ The structure of these graphs is still combinatorially rich, making reasoning about acyclic orientations challenging, but at the same time significantly less general than for arbitrary sets of NB constraints. Examples of such induced colorful graphs can be found in Fig.~\ref{fig:example-colorful-graph-for-approval}. For brevity, we henceforth assume that we have already checked in polynomial time whether there are any monochromatic cycles and rejected the input if so. In this section we will prove that if the colorful graph has passed this test, then an acyclic orientation is guaranteed to exist and can be determined in polynomial time. To show this, we prove a number of structural results about inducible colorful graphs, centering around the concept of \emph{biclique colors}, which are colors whose corresponding edges can be described by a cartesian product $A \times B$ for two disjoint sets $A, B \subseteq V.$ A color is \emph{star biclique} if it is a biclique color with either $|A| = 1$ or $|B| = 1.$ The two concepts are illustrated in Fig.~\ref{fig:example-colorful-graph-for-approval}. Not all colors in the inducible colorful graphs are biclique, as shown in Fig.~\ref{fig:example-colorful-graph-for-approval-b}. However, our result concerning them will be that whenever edges $a \edge b \edge c \edge a$ exist in the colorful graph and bear different colors, these three colors are biclique, and, moreover, the three cartesian products describing them are of the form $A \times B,$ $B \times C$ and $C \times A.$ In Fig.~\ref{fig:example-colorful-graph-for-approval-d} this can be observed for the triangle $1 \edge 2 \edge 3 \edge 1.$ Using this result, together with an analogue for star biclique colors, we then prove that if some orientation of the colorful graph has a cycle, then the smallest such cycle is of length three, showing that, intuitively, triangles are all one needs to worry about. To complete the proof of the main result, we then show that non-monochromatic triangles neccesarily consist of three different colors, implying that we can ignore all colors taking part in no three-colored triangles and focus on the rest. Afterward, a direct argument reusing our theorem about three-colored triangles can be used to construct an acyclic orientation of the colorful graph in polynomial time, proving the main claim of the paper.

\subsection{Biclique Colors} \label{sect:biclique}

From now on, all colorful graphs we consider are induced by approval preferences, unless stated otherwise. In this section, the goal is the prove that whenever $a \edge b \edge c \edge a$ is a three-colored triangle in the colorful graph, then these three colors are biclique and their edges are given by $A \times B,$ $B \times C$ and $C \times A$ for some disjoint sets $A, B, C \subseteq V$ where $a \in A, b \in B, c \in C.$ Following up, we then also give an instrumental analog for star biclique colors, namely that if $a \edge b \edge c$ are different-colored edges in the colorful graph, and edge $a \edge c$ does not exist in the colorful graph, then there is a set $B \subseteq V \setminus \{a, c\}$ with $b \in B$ such that the two colors consist of edges $\{a\} \times B$ and $B \times \{c\}$ respectively. This can be intuitively thought of as the case $|A| = |C| = 1$ of the result for three-colored triangles. Henceforth, unless stated otherwise, the preference profiles $\calP$ that we consider consist of approval ballots, and are represented through approval matrices.

Our high-level argument will hinge on a number of lower-level results concerning the effect that small local structures in the approval matrix have on the formula and colorful graphs. Below is the first such result that we will need. Its proof is a combination of direct reasoning with case analysis over $4 \times 6$ partially filled-in approval matrices. While we have considerably simplified the proof over naive case exhaustion, we found the details inessential, so left them for the appendix.


\begin{figure}[t]
    \centering
    \begin{subfigure}{.25\textwidth}
    \centering
    \includegraphics[width=.95\linewidth]{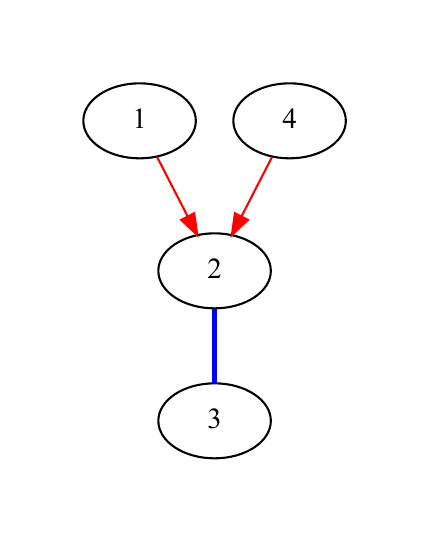}
    \caption{Assumed colors.}
    \label{fig:lemma-middle-a}
    \end{subfigure}
    \begin{subfigure}{.25\textwidth}
    \centering
    \includegraphics[width=.95\linewidth]{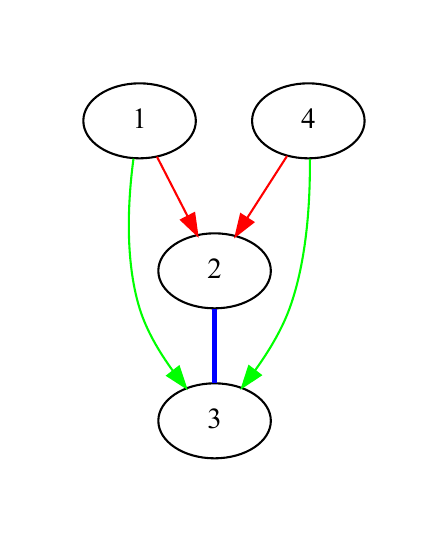}
    \caption{Implied colors.}
    \label{fig:lemma-middle-b}
    \end{subfigure}
    \caption{Illustration of Lemma \ref{lemma-middle}. Fig.~\ref{fig:lemma-middle-a} depicts the assumptions:~edge $(1, 2) \edge (4, 2)$ exists in the formula graph (red) and edge $(2, 3)$ exists in the colorful graph and bears a different color (blue). Additionally, edge $(1, 3)$ is neither red nor blue. Fig.~\ref{fig:lemma-middle-b} shows the edge colors implied by the setup, where green is a color different from red and blue.
    }
    \label{fig:lemma-middle}
\end{figure}

\begin{replemma}{lemma-middle} Consider four voters, say $1, 2, 3, 4,$ such that the formula graph contains the edge $(1, 2) \edge (4, 2)$ and that edge $(2, 3)$ exists in the colorful graph and has a different color than $(1, 2).$ If edge $(1, 3)$ is either not present in the colorful graph, or it has a different color than $(1, 2)$ and $(2, 3)$, then edge $(1, 3) \edge (4, 3)$ is guaranteed to be in the formula graph. The scenario is illustrated in Fig.~\ref{fig:lemma-middle}.
\end{replemma}

It is crucial to notice at this point that Lemma \ref{lemma-middle}, as well as all other results of a similar flavor that we will prove, continue to hold if some of the involved voters coincide. This can be seen by imagining cloning the recurring voters and applying the result with distinct voters. Armed with the previous lower-level result, we are now ready to prove the main result of the section, concerning the biclique structure of colors involved in three-colored triangles.

\begin{figure}[t]
    \centering
    \begin{subfigure}{.35\textwidth}
        \centering
        \includegraphics[width=.8\linewidth]{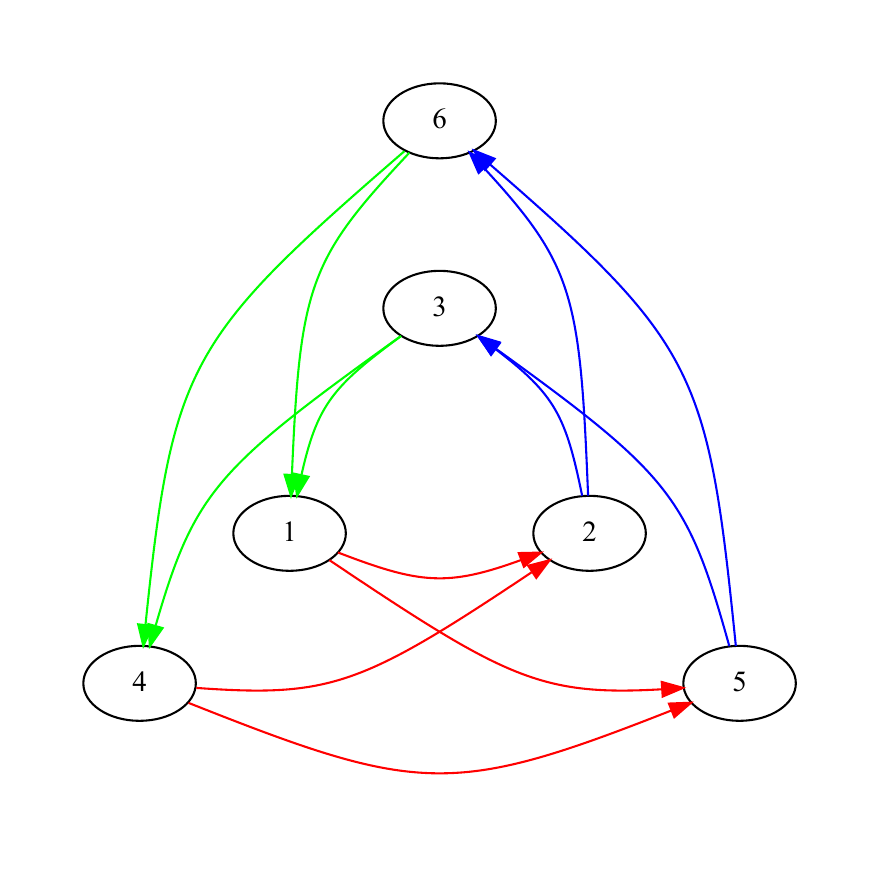}
        \caption{Col.~graph for $(S_1, S_2, S_3).$}
        \label{fig:lemma-triangle-cliques-a}
    \end{subfigure}
    \begin{subfigure}{.31\textwidth}
        \centering
        \includegraphics[width=.8\linewidth]{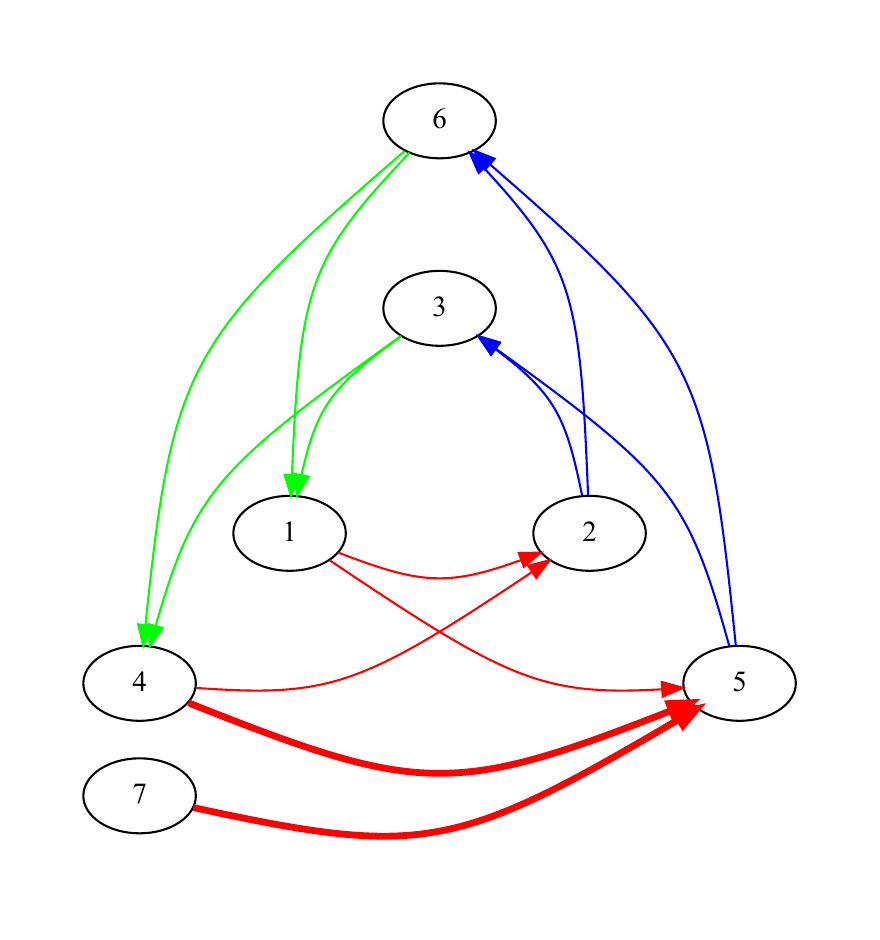}
        \caption{Assume $(4, 5, 7) \in \calC_\calP.$}
        \label{fig:lemma-triangle-cliques-b}
    \end{subfigure}
    \begin{subfigure}{.32\textwidth}
        \centering
        \includegraphics[width=.8\linewidth]{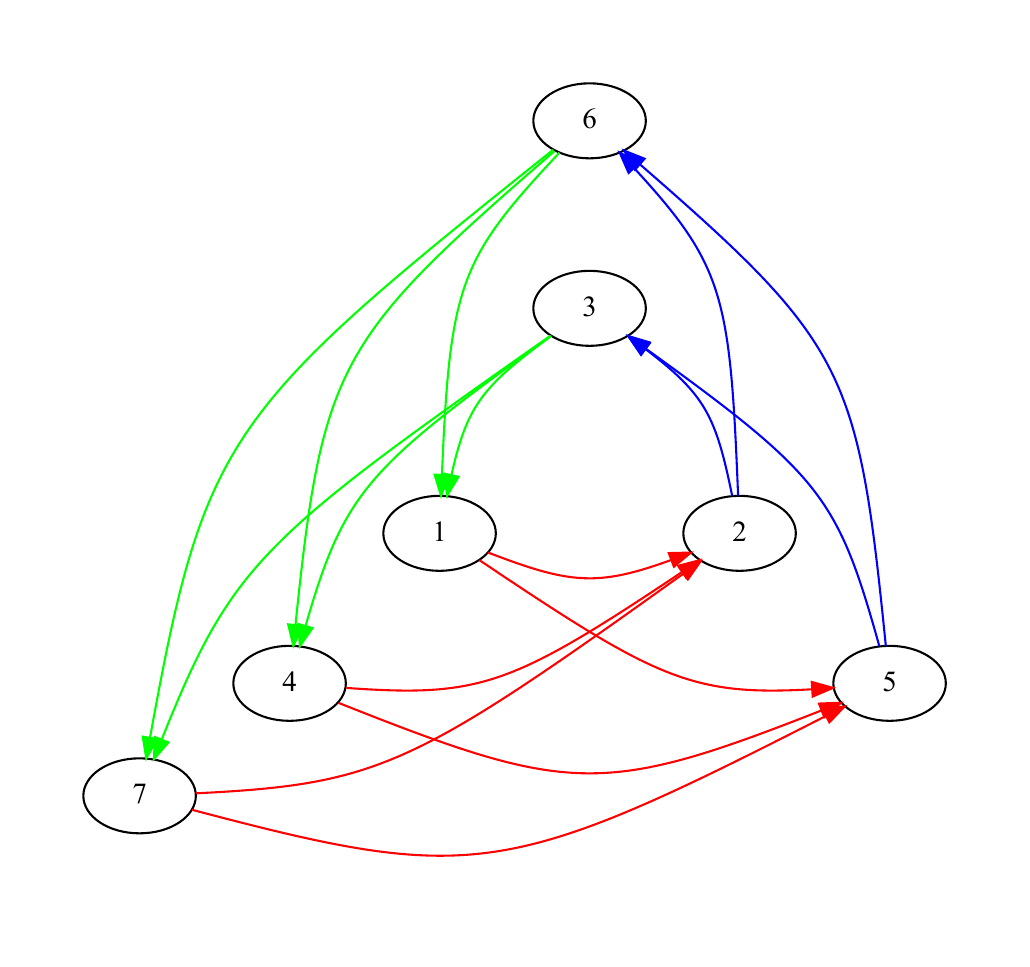}
        \caption{Contradicts maximality.}
        \label{fig:lemma-triangle-cliques-c}
    \end{subfigure}
    \caption{Illustration of the proof of Theorem \ref{lemma-triangle-cliques}. 
    We consider an inclusion-maximal triple $(S_1, S_2, S_3)$ such that $1 \in S_1, 2 \in S_2, 3 \in S_3$ and $S_1 \times S_2,$ $S_2 \times S_3$ and $S_3 \times S_1$ are connected in the formula graph, in this case $S_1 = \{1, 4\}, S_2 = \{2, 5\}$ and $S_3 = \{3, 6\}$ (Fig.~\ref{fig:lemma-triangle-cliques-a}). Then, we assume for a contradiction that these connected sets are not connected components, in this case because of NB constraint $(4, 5, 7)$ (Fig.~\ref{fig:lemma-triangle-cliques-b}). 
    Finally, this implies that $(S_1 \cup \{7\}, S_2, S_3)$ also has the property, contradicting 
    maximality (Fig.~\ref{fig:lemma-triangle-cliques-c}).
    }
    \label{fig:lemma-triangle-cliques}
\end{figure}

\begin{theorem} \label{lemma-triangle-cliques} Consider three voters, say $1, 2$ and $3,$ such that edges $(1, 2),$ $(2, 3)$ and $(3, 1)$ exist in the colorful graph and have different colors, then there exist disjoint sets of voters $S_1,$ $S_2$ and $S_3$ such that $1 \in S_1$, $2 \in S_2$ and $3 \in S_3$ and the connected components of $(1, 2),$ $(2, 3)$ and $(3, 1)$ in the formula graph are $S_1 \times S_2,$ $S_2 \times S_3$ and $S_3 \times S_1$ respectively.
\end{theorem}
\begin{proof} Our proof strategy is summarized in Fig.~\ref{fig:lemma-triangle-cliques}. Consider an inclusion-maximal triple of disjoint sets $(S_1, S_2, S_3)$ such that $1 \in S_1$, $2 \in S_2,$ $3 \in S_3$ and the subgraphs of the formula graph induced by $S_1 \times S_2,$ $S_2 \times S_3$ and $S_3 \times S_1$ are connected. By inclusion-maximal we mean that for every $v \notin (S_1 \cup S_2 \cup S_3)$ the triples $(S_1 \cup \{v\}, S_2, S_3),$ $(S_1, S_2 \cup \{v\}, S_3)$ and $(S_1, S_2, S_3 \cup \{v\})$ do not satisfy the property. Note that the connectivity assumption implies that in the colorful graph edges in $S_1 \times S_2$ have the color of $(1, 2)$, edges in $S_2 \times S_3$ have the color of $(2, 3)$ and edges in $S_3 \times S_1$ have the color of $(3, 1)$. Denote by $C_{12}, C_{23}$ and $C_{31}$ the connected components of $(1, 2),$ $(2, 3)$ and $(3, 1)$ in the formula graph. We will now show that $C_{12} = S_1 \times S_2,$ $C_{23} = S_2 \times S_3$ and $C_{31} = S_3 \times S_1$. Assume for a contradiction that this was not the case, then, without loss of generality, $C_{12} \neq S_1 \times S_2$. Together with the above, this means that $S_1 \times S_2 \subsetneq C_{12}.$ Because $C_{12}$ induces a connected subgraph in the formula graph, and $\varnothing \neq S_1 \times S_2 \subsetneq C_{12},$ it follows that there is an edge in the formula graph crossing the cut $(S_1 \times S_2, C_{12} \setminus (S_1 \times S_2))$. Without loss of generality, this edge is of the form $(v_1, v_2) \edge (x, v_2),$ where $v_1 \in S_1, v_2 \in S_2$ and $x \notin S_1$.

Now, we show that for every $v_3 \in S_3$ there is an edge $(v_1, v_3) \edge (x, v_3)$ in the formula graph. Note that this means that all edges of the form $(x, v_3)$ for $v_3 \in S_3$ have the color of $(1, 3)$. Consider an arbitrary $v_3 \in S_3$, and note that we can instantiate Lemma \ref{lemma-middle} with $1 \mapsto v_1, 2 \mapsto v_2, 3 \mapsto v_3$ and $4 \mapsto x$ because $(v_1, v_2) \edge (x, v_2)$ is in the formula graph and $(v_1, v_2),$ $(v_2, v_3)$ and $(v_3, v_1)$ bear the colors of $(1, 2),$ $(2, 3)$ and $(3, 1)$, which are different. The instantiation gives us that there is an edge $(v_1, v_3) \edge (x, v_3)$ in the formula graph, as required.

From the above, we know that the edge $(v_1, 3) \edge (x, 3)$ exists in the formula graph. Similarly, we now show that for every $v_2' \in S_2$ there is an edge $(v_1, v_2') \edge (x, v_2')$ in the formula graph. This is done analogously, by invoking Lemma \ref{lemma-middle} with $1 \mapsto v_1, 2 \mapsto 3, 3 \mapsto v_2'$ and $4 \mapsto x$, which is sound because $(v_1, 3) \edge (x, 3)$ is in the formula graph and $(v_1, 3),$ $(3, v_2')$ and $(v_2', v_1)$ bear the colors of $(1, 3),$ $(3, 2)$ and $(2, 1)$, which are different. 

At this point, we have shown that $(S_1 \cup \{x\}) \times S_2$ and $S_3 \times (S_1 \cup \{x\})$ both induce connected subgraphs in the formula graph; in particular, all edges from $x$ to $S_2$ have the color of $(1, 2)$ and all edges from $x$ to $S_3$ have the color of $(1, 3)$. We will now additionally show that $x \notin (S_2 \cup S_3)$, from which the triple $(S_1 \cup \{x\}, S_2, S_3)$ also satisfies the property we started with, contradicting the hypothesis that $(S_1, S_2, S_3)$ was inclusion-maximal. First, if $x \in S_3,$ then $(1, 2)$ and $(2, 3)$ immediately have the same color, which cannot be the case. Otherwise, if $x \in S_2,$ then from our argument $(x, 3)$ has the same color as $(1, 3)$, while the assumptions tell us that it has the same color as $(2, 3)$, meaning that $(1, 3)$ and $(2, 3)$ share the same color, which is again not possible.
\end{proof}

The following result can be seen as a star biclique analogue of Theorem \ref{lemma-triangle-cliques}. Intuitively, one can think of this as the orthogonal case where the edge $1 \edge 3$ does not exist in the colorful graph and $S_1 = \{1\}, S_3 = \{3\}.$ While the two proofs are very similar in spirit, we could not find a transparent way of unifying the two arguments, partly because of less symmetry in the scenario. In particular, this lack of symmetry will require a second low-level result in the spirit of Lemma \ref{lemma-middle}, stated and proven in the appendix. See the appendix for the relevant proofs.

\begin{reptheorem}{lemma-no-edge-cliques} Consider three voters, say $1, 2$ and $3,$ such that edges $(1, 2)$ and $(2, 3)$ exist in the colorful graph and have different colors, while edge $(1, 3)$ does not exist in the colorful graph. Then, there exists a set of voters $S$ such that $2 \in S$ and $1, 3 \notin S$ with the property that the connected components of $(1, 2)$ and $(2, 3)$ in the formula graph are $\{1\} \times S$ and $S \times \{3\}$ respectively.
\end{reptheorem}

\subsection{Only Three-Colored Triangles Matter, the Main Result}\label{sec:only-triangles-matter}

In this section we prove that if some orientation of the colorful graph has a cycle, then the smallest such cycle is a triangle. Then, we show that non-monochromatic triangles neccesarily consist of three different colors. Finally, we prove that any colorful graph with no monochromatic cycles has an acyclic orientation, hinging on Theorem \ref{lemma-triangle-cliques}.

\begin{lemma} \label{lemma-only-triangles-matter-for-orientation} An orientation of the colors in the colorful graph is acyclic if and only if it induces no directed cycles of length three.
\end{lemma}
\begin{proof} The ``only if'' direction is straightforward. In what follows we prove the ``if'' direction. Consider an arbitrary orientation of the colors in the colorful graph. 
Cycles of length one and two are not possible because the colorful graph lacks self-loops and parallel edges, so assume for a contradiction that the length of the shortest cycle induced by our orientation is at least four. Let such a shortest cycle be $C = v_1 \rightarrow v_2 \rightarrow \ldots \rightarrow v_k \rightarrow v_1,$ where $k \geq 4.$ Observe that the minimum-length assumption implies that $C$ is an induced cycle, because any extra edge in the subgraph of the colorful graph induced by $\{v_1, \ldots, v_k\}$ would imply a shorter cycle. If all edges in $C$ are of the same color, this would contradict the fact that all colors are acyclic when seen individually, so $C$ consists of at least two distinct colors, so at least two distinct colors occur along $C.$

We now prove that all colors $c$ occurring along $C$ are star biclique colors. To do so, for any such $c$ note that there exists an index $1 \leq i \leq k$ such that $v_i \xrightarrow[]{c} v_{i + 1}$ and $v_{i + 1} \xrightarrow[]{c'} v_{i + 2},$ where $c' \neq c$ is some color other than $c$ and index arithmetic is performed modulo $k$ such that $v_{k + 1} = v_1$, etc. Because $C$ is an induced cycle, it follows that edge $v_i \edge v_{i + 2}$ does not appear in the colorful graph. Since additionally $c \neq c',$ instantiating Theorem \ref{lemma-no-edge-cliques} with $1 \mapsto v_k, 2 \mapsto v_{k + 1}, 3 \mapsto v_{k + 2}$ gives that $c$ and $c'$ are star biclique colors. 

Since all colors occurring along $C$ are star biclique, every color can only occur in $C$ at most once, as any second occurrence would imply that the common node of the edges in a star is reused, meaning $C$ is not a simple cycle, and hence not minimal. Therefore, let $c_1, c_2, \ldots, c_k$ be the pairwise distinct colors of the edges in $C$, in order.

Now, applying Theorem \ref{lemma-no-edge-cliques} with $1 \mapsto v_1, 2 \mapsto v_2, 3 \mapsto v_3$ gives that the connected component of $(v_1, v_2)$ in the formula graph is $\{v_1\} \times S$ and that of $(v_2, v_3)$ is $S \times \{v_3\},$ for some $S$ such that $v_1, v_3 \notin S.$ Similarly, applying Theorem \ref{lemma-no-edge-cliques} with $1 \mapsto v_2, 2 \mapsto v_3, 3 \mapsto v_4$ gives that the connected component of $(v_2, v_3)$ in the formula graph is $\{v_2\} \times S'$ and that of $(v_3, v_4)$ is $S' \times \{v_4\},$ for some $S'$ such that $v_2, v_4 \notin S'.$ As a result, the connected component of $(v_2, v_3)$ is at the same time $S \times \{v_3\}$ and $\{v_2\} \times S'.$ Setting the two equal yields $S \times \{v_3\} = \{v_2\} \times S',$ from which the connected component is actually exactly $\{(v_2, v_3)\},$ contradicting the fact that $(v_2, v_3)$ occurs in the colorful graph.
\end{proof}

Next, we prove that only three-colored triangles need to be considered; e.g., the case in Fig.~\ref{fig:example-for-intro-colorful-graph} cannot occur. For this we 
will
need the following stronger form of Lemma \ref{lemma-middle}:

\begin{replemma}{lemma-stronger-middle}[Strengthened Lemma \ref{lemma-middle}] Consider four voters, say $1, 2, 3, 4,$ such that the formula graph contains the edge $(1, 2) \edge (4, 2)$ and that edge $(2, 3)$ exists in the colorful graph. If the formula graph does not contain edges $(3, 2) \edge (1, 2)$ and $(3, 2) \edge (4, 2),$ then the formula graph is guaranteed to contain the edge $(1, 3) \edge (4, 3).$ 
\end{replemma}

The proof closely follows that of Lemma \ref{lemma-middle}, but the weaker assumptions make it significantly more difficult to proceed as before in a principled manner. Instead, we resort to a computer program to perform the case analysis (see appendix for details).

\begin{lemma} \label{lemma-no-bicolored-cycle-triangles} Consider three voters, say $1, 2, 3$ such that edges $1 \edge 2$ and $2 \edge 3$ are the same color and oriented $1 \rightarrow 2$ and $2 \rightarrow 3$ in the colorful graph. If edge $1 \edge 3$ appears in the colorful graph, then it is of the same color as $1 \edge 2 \edge 3.$
\end{lemma}
\begin{proof} Without loss of generality, assume that $1 \edge 3$ is oriented $1 \rightarrow 3$ in the colorful graph. Assume for a contradiction that the color of $1 \edge 2 \edge 3$ is $c$ and the color of $1 \edge 3$ is $c' \neq c.$ Because $1 \xrightarrow{c} 2$ and $2 \xrightarrow{c} 3$ are the same color, it follows that there is a simple path in the formula graph $(1, 2) = (a_1, b_1) \edge (a_2, b_2) \edge \ldots \edge (a_k, b_k) = (2, 3).$ We will prove by induction on $1 \leq \ell \leq k$ that $a_\ell \xrightarrow{c'} 3$ and $b_\ell \xrightarrow{c} 3,$ which is a contradiction since $b_k = 3$ and there are no self-loops in the colorful graph. The base case $\ell = 1$ is true by assumption. Assume that $a_\ell \xrightarrow{c'} 3$ and $b_\ell \xrightarrow{c} 3$ for some $\ell < k$, we then want to prove that $a_{\ell + 1} \xrightarrow{c'} 3$ and $b_{\ell + 1} \xrightarrow{c} 3.$ By construction of the formula graph, distinguish two cases:

If $a_{\ell + 1} = a_{\ell},$ then we instantiate Lemma \ref{lemma-stronger-middle} with $1 \mapsto b_{\ell}, 2 \mapsto a_{\ell}, 3 \mapsto 3, 4 \mapsto b_{\ell + 1}.$ This is sound because edge $(b_{\ell}, a_{\ell}) \edge (b_{\ell + 1}, a_{\ell})$ is in the formula graph and edges $(3, a_{\ell}) \edge (b_{\ell}, a_{\ell})$ and $(3, a_{\ell}) \edge (b_{\ell + 1}, a_{\ell})$ are not, since $3 \edge a_{\ell}$ is color $c'$, while the other two are color $c$ by assumption.
The instantiation gives that edge $(b_\ell, 3) \edge (b_{\ell + 1}, 3)$ is in the formula graph, from which $b_{\ell + 1} \xrightarrow{c} 3.$

If $b_{\ell + 1} = b_{\ell},$ then we instantiate Lemma \ref{lemma-stronger-middle} with $1 \mapsto a_{\ell}, 2 \mapsto b_{\ell}, 3 \mapsto 3, 4 \mapsto a_{\ell + 1}.$ This is sound because edge $(a_{\ell}, b_{\ell}) \edge (a_{\ell + 1}, b_{\ell})$ is in the formula graph and edges $(3, b_{\ell}) \edge (a_{\ell}, b_{\ell})$ and $(3, b_{\ell}) \edge (a_{\ell + 1}, b_{\ell})$ are not, since the colorful graph has edges $a_{\ell} \xrightarrow{c} b_\ell,$ $a_{\ell + 1} \xrightarrow{c} b_\ell$ and $b_{\ell} \xrightarrow{c} 3,$ which together with any of $(3, b_{\ell}) \edge (a_{\ell}, b_{\ell})$ or $(3, b_{\ell}) \edge (a_{\ell + 1}, b_{\ell})$ in the formula graph would lead to a path between $(b_{\ell}, 3)$ and $(3, b_{\ell})$ in the formula graph, which cannot be the case. The instantiation gives us that the edge $(a_\ell, 3) \edge (a_{\ell + 1}, 3)$ exists in the formula graph, from which $a_{\ell + 1} \xrightarrow{c} 3.$
\end{proof}


\begin{corollary} \label{corollary-only-tricolored-triangles-matter} An orientation of the colors in the colorful graph is acyclic if and only if it induces no directed cycles of length three consisting of three different colors.
\end{corollary}


\begin{theorem} \label{theorem-acyclic-orientation-exists} An acyclic orientation of the colorful graph exists and can be computed in poly-time.
\end{theorem}
\begin{proof} By Corollary \ref{corollary-only-tricolored-triangles-matter}, only colors participating in three-colored triangles have to be considered (all other colors we can orient arbitrarily). Let $C = \{c_1, \ldots, c_k\}$ be the set of such colors. For each color $c_i \in C$, applying Theorem \ref{lemma-triangle-cliques} to an arbitrary three-colored triangle involving $c_i$ gives that $c_i$ is a biclique color, so the set of edges colored with color $c_i$ in the colored graph is $A_i \times B_i$ for some disjoint sets of vertices $A_i, B_i$.
For a finite set of numbers $S$, we write $\min S$ for the minimum number in $S$. We now construct an orientation of the colors in $C$ as follows:~for color $c_i \in C$ we orient the edges of color $c_i$ to go $A_i \rightarrow B_i$ if $\min A_i < \min B_i$ and $B_i \rightarrow A_i$ otherwise. Note that $\min A_i \neq \min B_i$ as $A_i \cap B_i = \varnothing.$ We now prove that this orientation is acyclic to complete the proof. Assume for a contradiction that the orientation induced a cycle $a \rightarrow b \rightarrow c \rightarrow a$. Then, by Corollary \ref{corollary-only-tricolored-triangles-matter}, the three edges involved in the cycle bear different colors in the colorful graph, so, by Theorem \ref{lemma-triangle-cliques}, there exist three disjoint vertex sets $A, B, C$ such that $a \in A, b \in B, c \in C$ and the connected components of $a \rightarrow b,$ $b \rightarrow c$ and $c \rightarrow a$ in the formula graph are given by $A \times B,$ $B \times C$ and $C \times A$. However, by construction of the orientation, this would mean that $\min A < \min B < \min C < \min A$, a contradiction.
\end{proof}

Armed as such, we can now prove the main result of our paper.

\begin{theorem}\label{th:main} There is a polynomial time algorithm taking as input a collection of approval ballots that outputs a profile of single-crossing preferences that could have lead to the given approval votes. If this is not possible, the algorithm reports accordingly.
\end{theorem}

\begin{proof} Start with the approval ballots and build the formula graph. If there are pairs of complementary nodes in the same connected component, report impossibility. Otherwise, build the colorful graph. If a monochromatic cycle exists, report impossibility. Otherwise, use Theorem \ref{theorem-acyclic-orientation-exists} to compute an acyclic orientation, and extend it to a linear order $\triangleleft.$ Finally, use the approach of Elkind et al.~\cite{edith_incomplete_sc}, explained in Appendix \ref{app:equiv-ssc-psc}, to compute a single-crossing profile with respect to the known axis $\triangleleft.$
\end{proof}

\section{Conclusions and Future Work}

We gave a polynomial time algorithm computing a collection of fully-ranked ballots that could have been induced by the given approval ballots. Furthermore, we showed that there is no finite substructures characterization of possibly single-crossing elections (deferred to Appendix \ref{sect:impossibility}). To prove our main result, we developed a new algorithm for the non-betweenness problem, showing that it is FPT with respect to the number of colors in a certain colorful graph. We expect this FPT algorithm to lead to polynomial algorithms for other ordering problems reducing to the non-betweenness problem by combinatorial analysis of the induced colorful graphs. 
We note that betweenness reduces to non-betweenness since the constraint ``$b$ has to be between $a$ and $c$'' can be modeled as the conjunction of two non-betweenness constraints, namely: ``$a$ should not be between $b$ and $c$'' and ``$c$ should not be between $a$ and $b$'', so our framework can also be used for betweenness.

As steps for future research, the most natural next question to ask is the complexity of recognizing possibly single-crossing for multi-valued approval, the case of three indifference classes being a natural first candidate to consider. We have empirically found that our proof of correctness no longer applies in this case, since for instance Lemma \ref{lemma-middle} fails to hold for 3-valued approval. However, we could not find instances where the algorithm as a whole fails the recognition task. We leave it open to settle this question. Another promising avenue for further investigation stands in the case when our algorithm reports that voters' innate preferences are not single-crossing. In this case, one might hope that they are at least somewhat close to being single-crossing, so it would be interesting to investigate how to recover preferences that are nearly single-crossing from approval ballots (see \cite[Section 4.8]{survey_restricted} for a discussion of related results).

\subsubsection{Acknowledgements} We thank Finn Harbeke and Edith Elkind for the many fruitful discussions regarding this work. We moreover thank Finn Harbeke for checking whether our proposed approach works for 3-valued approval and making suggestions to improve the paper. We thank the anonymous reviewers for their constructive feedback and useful suggestions contributing to improving the paper.

{\footnotesize
\bibliography{arxiv}}

\newpage

\appendix

\begin{figure}[t]
    \centering
    \begin{subfigure}{.45\textwidth}
    \centering
    \begin{tabular}{p{\matrixcolumnwidth} | *{5}{p{\matrixcolumnwidth}}}
          & \hcr{1} & \hcr{2} & \hcr{3} & \hcr{4} & \hcr{5} \\
          \hline
        \hcb{1} & \textbf{1} & \textbf{1} & 0 & 0 & 0 \\
        \hcb{2} & 0 & \textbf{1} & \textbf{1} & 0 & 0 \\
        \hcb{3} & 0 & 0 & \textbf{1} & \textbf{1} & 0 \\
        \hcb{4} & 0 & 0 & 0 & \textbf{1} & \textbf{1} \\
        \hcb{5} & \textbf{1} & 0 & 0 & 0 & \textbf{1} \\
    \end{tabular}
    \vspace{.56cm} 
    \caption{Approval matrix of $\calP_5.$}
    \label{fig:matrix-impossible-profile}
    \end{subfigure}
    \begin{subfigure}{.5\textwidth}
        \centering
        \includegraphics[width=.6\linewidth]{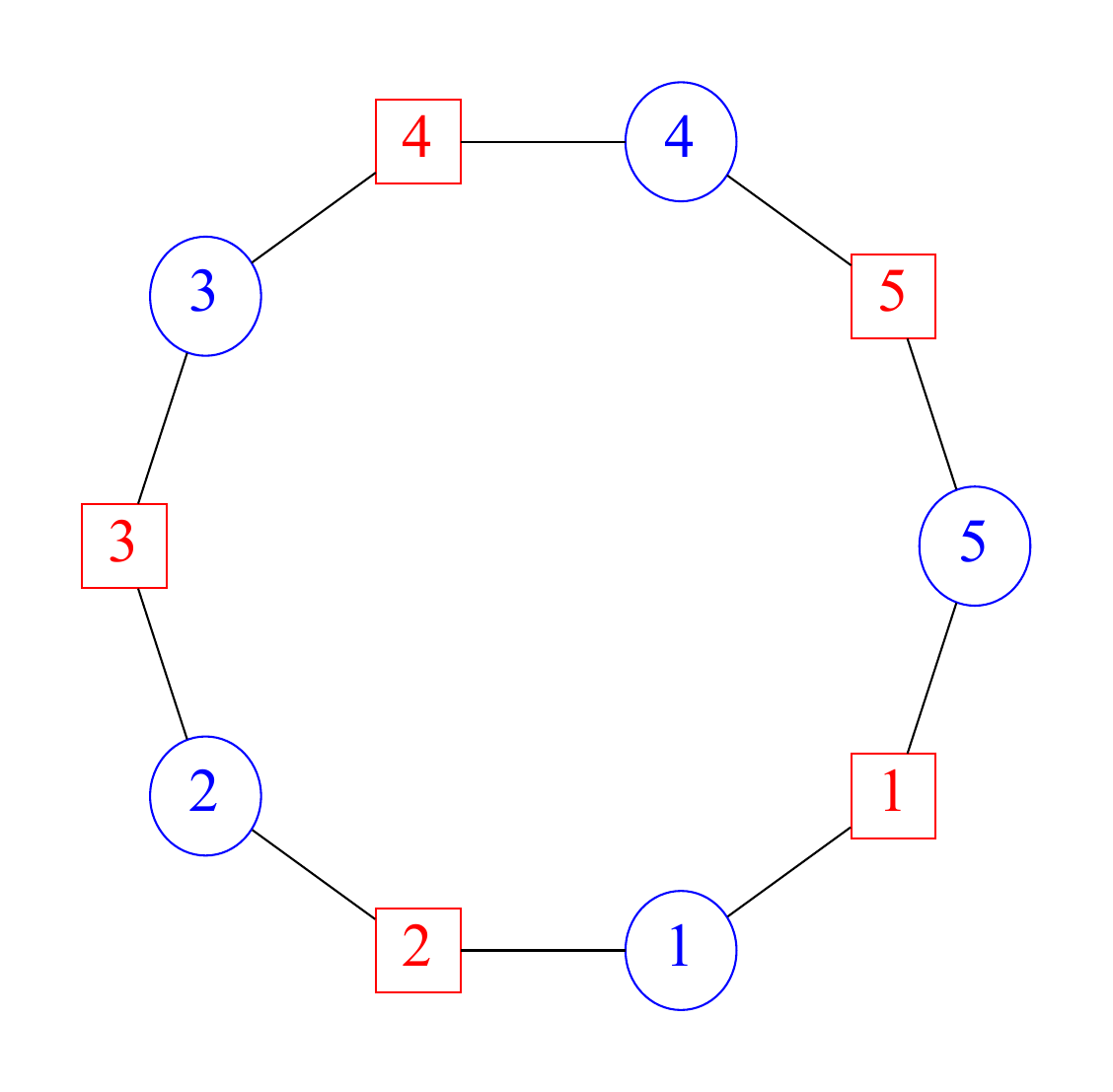}
        \caption{Bipartite graph representation of $\calP_5$.}
        \label{fig:bipartite-graph-of-impossible-profile}
    \end{subfigure}
    \caption{Construction used in the impossibility proof for $n = 5$ voters/alternatives.
    Fig.~\ref{fig:matrix-impossible-profile} shows the approval matrix. Fig.~\ref{fig:bipartite-graph-of-impossible-profile} shows the corresponding bipartite preference graph, with voters represented by red squares and candidates by blue circles.}
    \label{fig:impossible-profile}
\end{figure}

\section{Impossibility of Finite Forbidden Substructures Characterization} \label{sect:impossibility}

Single-crossing elections can be characterized as those elections not containing as subprofiles two specific elections consisting of three voters and six alternatives, and four voters and four alternatives \cite{characterization_minors_single_crossing}.
In this appendix, we show that, in contrast,
possibly single-crossing elections do not admit such a characterization in terms of finitely many forbidden subprofiles. To do so, for $n \geq 2$ define an approval profile $\calP_n$ where $n$ voters cast approval ballots over $m = n$ alternatives, voter $i \in [n]$ approving of candidates $i$ and $i - 1$ and disapproving of the rest (Fig.~\ref{fig:matrix-impossible-profile}).\footnote{Here we assume that addition and subtraction are performed modulo $n,$ using the fact that $0 \equiv n \pmod n$} A useful notion for understanding these profiles is that of the \emph{preference graph} of an approval profile, obtained by introducing a vertex for each voter and a vertex for each candidate, together with undirected edges drawn between voters and the candidates that they approve of.\footnote{This is the same as interpreting the approval matrix as a bipartite adjacency matrix.} Subprofiles of a profile then correspond to induced subgraphs of the preference graph respecting node types. In the case of $\calP_n,$ the preference graph is a cycle (Fig.~\ref{fig:bipartite-graph-of-impossible-profile}). We call a profile that is not seemingly single-crossing \emph{minimal} if it has no non-seemingly-single-crossing proper subprofiles. The following shows that $(\calP_n)_{n \geq 4}$ is an infinite family of minimal not seemingly single-crossing profiles, implying that this family has to be part of any forbidden substructures characterization, and hence the conclusion.

\begin{replemma}{impossible-profile-is-not-ssc} For $n \geq 4$ the profile $\calP_n$ is not seemingly single-crossing. Moreover, any profile formed from $\calP_n$ by removing voters and/or candidates is seemingly single-crossing.
\end{replemma}
\begin{proof} To show the first part, fix some $n \geq 4$ and for brevity write $\calP = \calP_n$. We proceed by considering the set of NB constraints $\calC_{\calP}$ induced by $\calP$. We first look at voter 1. Consider another voter $i \notin \{1, 2, n\}.$ Note that voter 1 approves of candidate 1, but disapproves of candidate $i,$ while voters $i$ and $i + 1$ disapprove of candidate 1, but approve of candidate $i.$ Consequently, we get the NB constraints $(i, 1, i + 1)$ for $i \notin \{1, 2, n\}.$ Additionally, voter 1 approves of candidate $n,$ but disapproves of candidate $2,$ while voters 2 and 3 do the opposite,\footnote{No longer true for $n \leq 3,$ from which the need for $n \geq 4$ becomes apparent. In fact, $\calP_2$ and $\calP_3$ are actually SSC.} so we also get the NB constraint $(2, 1, 3)$. Together, so far we have the NB constraints $(i, 1, i + 1)$ for $i \notin \{1, n\}$. In other words, voter 1 cannot be between any other two voters, so has to go either first or last in order for $\calC_{\calP}$ to be satisfied. By symmetry of the construction; i.e., see Fig.~\ref{fig:bipartite-graph-of-impossible-profile}; the same argument also applies to the other voters, but this cannot be the case, as at most two voters can go first/last in the order. To show the second part, by symmetry (see Fig.~\ref{fig:bipartite-graph-of-impossible-profile}) it is enough to consider the cases of removing voter 1 or candidate $n$. In either case, one can check that the linear order $1 \triangleleft 2 \triangleleft \ldots \triangleleft n$ witnesses the seemingly single-crossing property.
\end{proof}

\section{Equivalence of SSC and PSC for Weak Orders}\label{app:equiv-ssc-psc}

In this appendix, we show that a profile of weak orders is seemingly single-crossing (SSC) with respect to an axis $\triangleleft$ if and only if it is possibly single-crossing (PSC) with respect to $\triangleleft.$ The latter implies the former straightforwardly, so we focus on the non-trivial direction, namely that SSC with respect to $\triangleleft$ implies PSC with respect to $\triangleleft.$ To this end, we show that under the assumption that the profile is SSC with respect to $\triangleleft,$ we can extend each weak order into a linear order such that the resulting profile is single-crossing with respect to $\triangleleft.$ As is required by our main algorithm (Theorem \ref{th:main}), this can be done in polynomial time. The proof that we give is a streamlined version of the argument of Elkind et al.~\cite{edith_incomplete_sc}.

\begin{theorem} Given a preference profile of weak orders $\calP = (\succ_i)_{i \in [n]}$ and a linear order $\triangleleft$ over $[n]$ such that $\calP$ is seemingly single-crossing with respect to $\triangleleft,$ one can compute in polynomial time for each voter $i \in [n]$ a linear extension $\succ_i'$ of $\succ_i$ such that $(\succ_i')_{i \in [n]}$ is single-crossing with respect to $\triangleleft.$
\end{theorem}
\begin{proof} Recall that $C = [m]$ is the set of candidates/alternatives. First, if there exists a pair of distinct candidates $(a, b) \in C^2$ such that $a \approx_i b$ for all voters $i \in [n],$ then remove $b$ and apply the theorem recursively to the reduced instance. After getting an answer for the reduced instance, add back candidate $b$ into all voters' rankings right next to $a$ such that $a \succ_i' b$ holds for all voters $i \in [n].$ Hence, we may now assume that no such pair $(a, b)$ exists.\footnote{Note that if the input is presented in matrix form, the recursive steps that are performed for a given instance correspond to identifying and removing equal rows in the preference matrix, and then adding back the removed candidates at the end.} Without loss of generality, further assume that $1 \triangleleft \ldots \triangleleft n.$

We construct the linear extensions $(\succ_i')_{i \in [n]}$ of $(\succ_i)_{i \in [n]}$ by specifying for each pair of alternatives $(a, b) \in C^2$ with $1 \leq a < b \leq m$ and each voter $i \in [n]$ whether $a \succ_i' b$ or $b \succ_i' a.$ Consider such a pair $(a, b).$ For each voter $i \in [n]$ proceed as follows:~compute the highest $j$ such that $1 \leq j \leq i$ and $a \not\approx_j b.$ If no such $j$ exists, then take $j$ to be the smallest $j$ such that $i < j \leq n$ and $a \not\approx_j b.$ Note that such a $j$ has to exist since we assumed that for any pair of distinct alternatives, at least one voter is not indifferent between them. If $a \succ_j b,$ then we make $\succ_i'$ compare $a$ and $b$ as $a \succ_i' b.$ Conversely, if $b \succ_j a,$ then we make it compare $a$ and $b$ as $b \succ_i' a.$ In Iverson bracket notation, the two cases can be more succinctly expressed as setting $[a \succ_i' b] = [a \succ_j b].$ Implicitly, we also set $[b \succ_i' a] = [b \succ_j a].$ This completes the construction, which can be seen to only require polynomial time to compute.

Next, we show the correctness of the construction. First, observe that orders $(\succ_i')_{i \in [n]}$ are extensions of $(\succ_i)_{i \in [n]}.$ This is because if $a \not\approx_i b$ then $j = i$ will be chosen in the construction, implying the required. Next, let us show that $(\succ_i')_{i \in [n]}$ is single-crossing with respect to $\triangleleft.$ This amounts to showing that for any two distinct candidates $(a, b) \in C^2$ there are no three voters $i \triangleleft j \triangleleft k$ such that $a \succ_i' b,$ $b \succ_j' a$ and $a \succ_k' b.$ Equivalently, for any two candidates $(a, b) \in C^2$ with $a < b$ there are no three voters $i \triangleleft j \triangleleft k$ such that $a \succ_i' b,$ $b \succ_j' a,$ $a \succ_k' b,$ or $b \succ_i' a,$ $a \succ_j' b,$ $b \succ_k' a.$ To show this, consider the sequence $(s_i')_{i \in [n]}$ where $s_i' = [a \succ'_i b]$; i.e., $s_i' = 1$ if $a \succ'_i b$ and 0 if $b \succ'_i a.$ The statement then amounts to there not existing indices $1 \leq i < j < k \leq n$ such that $(s_i', s_j', s_k') \in \{(0, 1, 0), (1, 0, 1)\},$ which is in turn equivalent to $s'$ being either monotonically increasing or decreasing. We now show that this monotonicity property is indeed satisfied:~consider additionally the sequence $(s_i)_{i \in [n]}$ where $s_i \in \{\bot, 0, 1\}$ satisfies $s_i = \bot$ if $a \approx_i b$ and otherwise $s_i = [a \succ_i b]$ holds. Our construction of $(\succ'_i)_{i \in [n]}$ from $(\succ_i)_{i \in [n]}$ can then be reformulated as follows: for every pair of distinct candidates $(a, b) \in C^2$ with $a < b$ we construct $s'$ from $s$ by replacing $s_i = \bot$ entries by either 0 or 1. In particular, if $s_i = \bot,$ then we compute the last non-$\bot$ element to the left of $i$ inclusively, or the first non-$\bot$ element to the right of $i$ if none of the previous kind were found, and then setting $s_i'$ to that element. Because $\calP$ is SSC with respect to $\triangleleft,$ it follows that $s$ is either monotonically increasing or decreasing if we ignore $\bot$ elements, so this procedure completes $s$ into a sequence $s'$ without $\bot$ elements that is also monotonically increasing/decreasing, as required.

It remains to show that the relations  $(\succ'_i)_{i \in [n]}$ produced from $(\succ_i)_{i \in [n]}$ are indeed linear orders. By construction, they are already total, irreflexive, and antisymmetric, so it remains to show that they are also transitive. Assume for a contradiction that for some voter $i \in [n]$ order $\succ'_i$ is not transitive. Hence, there exist three candidates $a, b, c \in C$ such that $a \succ'_i b \succ'_i c \succ'_i a.$ We will show that this cannot be the case. To do so, consider how the values $[a \succ'_i b], [b \succ'_i c]$ and $[c \succ'_i a]$ were constructed. For ease of writing, define a linear $\sqsubset$ order on $[n]$ such that $i \sqsubset \ldots \sqsubset 1 \sqsubset i + 1 \sqsubset \ldots \sqsubset n.$ Then:

\vspace{3pt}

\noindent \textbullet\ $[a \succ'_i b] = [a \succ_{j_{\mathit{ab}}} b]$ where $j_{\mathit{ab}} = \min_{\sqsubset}\{j \mid a \not\approx_{j} b\}$; 

\noindent \textbullet\ $[b \succ'_i c] = [b \succ_{j_{\mathit{bc}}} c]$ where $j_{\mathit{bc}} = \min_{\sqsubset}\{j \mid b \not\approx_{j} c\}$;

\noindent \textbullet\ $[c \succ'_i a] = [c \succ_{j_{\mathit{ca}}} a]$ where $j_{\mathit{ca}} = \min_{\sqsubset}\{j \mid c \not\approx_{j} a\}$.

\vspace{3pt}

\noindent Without loss of generality, assume $\min_{\sqsubset}(j_{\mathit{ab}}, j_{\mathit{bc}}, j_{\mathit{ca}}) = j_{\mathit{ab}}.$ Recall that $\succ_{j_{\mathit{ab}}}$ is a weak order and note that $a \succ_{j_{\mathit{ab}}} b.$

If $a \approx_{j_{\mathit{ab}}} c,$ then $a \succ_{j_{\mathit{ab}}} b$ implies that $c \succ_{j_{\mathit{ab}}} b.$ Consequently, $j_{\mathit{bc}} = j_{\mathit{ab}},$ from which $c \succ'_i b,$ contradicting the assumption that $b \succ'_i c.$ As a result, $a \not\approx_{j_{\mathit{ab}}} c$ must hold.

If $b \approx_{j_{\mathit{ab}}} c,$ then $a \succ_{j_{\mathit{ab}}} b$ implies that $a \succ_{j_{\mathit{ab}}} c.$ Consequently, $j_{\mathit{ca}} = j_{\mathit{ab}},$ from which $a \succ'_i c,$ contradicting the assumption that $c \succ'_i a.$ As a result, $b \not\approx_{j_{\mathit{ab}}} c,$ must hold.

Since $a \not\approx_{j_{\mathit{ab}}} c$ and $b \not\approx_{j_{\mathit{ab}}} c$ both hold, it follows that $j_{\mathit{ab}} = j_{\mathit{bc}} = j_{\mathit{ca}},$ and so $\succ_{j_{\mathit{ab}}}$ restricted to $a, b, c$ is a total order, from which  $\succ_i'$ also is, contradicting the assumed violation of transitivity $a \succ'_i b \succ'_i c \succ'_i a.$
\end{proof}

\section{Proofs Omitted From the Main Body}

In this appendix, we provide the proofs missing from the main text. 
In some cases, the proofs will require stating and proving additional intermediate results.

\subsection{Proof of Lemma \ref{lemma:reduction_to_nb}}

\repeatlemma{lemma:reduction_to_nb}
\begin{proof} Assume $\calP$ is not SSC with respect to $\triangleleft.$ Then, there are voters $i \triangleleft j \triangleleft k$ and candidates $a, b \in C$ such that $a \succ_i b, b \succ_j a$ and $a \succ_k b$. By construction of $\calC_\calP,$ this means that $(i, j, k) \in \calC_\calP$, so $\calC_\calP$ is not satisfied by $\triangleleft.$ Conversely, assume $\calC_\calP$ is not satisfied by $\triangleleft.$ Then, there is $(i, j, k) \in \calC_\calP$ such that either $i \triangleleft j \triangleleft k$ or $k \triangleleft j \triangleleft i.$ By construction of $\calC_\calP,$ this means that $a \succ_i b, b \succ_j a, a \succ_k b$ and either $i \triangleleft j \triangleleft k$ or $k \triangleleft j \triangleleft i.$ Either way, this means that $\calP$ is not SSC with respect to $\triangleleft$.
\end{proof}

\subsection{Proof of Lemma \ref{lemma_boolean_formula_encodes_ordering}}

Recall the following basic fact:

\begin{proposition}\label{prop-iverson} A NB constraint $(i, j, k)$ is satisfied by a linear order $\triangleleft$ iff $[i \triangleleft j] = [k \triangleleft j].$
\end{proposition}

\repeatlemma{lemma_boolean_formula_encodes_ordering}

\begin{proof} Assume linear order $\triangleleft$ satisfies the constraints in $\calC$, then let us set $x_{i, j} = [i \triangleleft j]$. As a result, constraints in \ref{bool_form_2} are satisfied due to Proposition \ref{prop-iverson}, while constraints in \ref{bool_form_1} and \ref{bool_form_3} are satisfied because $\triangleleft$ is antisymmetric and transitive. Conversely, assume values $x_{i, j}$ are a satisfying assignment of $\Phi_\calC$. Then, one can define $\triangleleft$ by setting $[i \triangleleft j] = x_{i, j},$ which is a linear order because \ref{bool_form_1} and \ref{bool_form_3} enforce antisymmetry and transitivity. Moreover, $\triangleleft$ satisfies the constraints in $\calC$ because $\Phi_\calC$ satisfies constraint set \ref{bool_form_2}, by using Proposition \ref{prop-iverson}.
\end{proof}

\subsection{Proof of Lemmas \ref{lemma_connected} and \ref{lemma_complement_is_cc}}

\repeatlemma{lemma_connected}
\begin{proof} Assume $x$ is an assignment satisfying set \ref{bool_form_2_reformulated}. Let $S$ be a connected component of $G_\calC$ and $u, v \in S$ be arbitrary. Since $S$ is a connected component, consider a path $u = t_0\edge t_1\edge \ldots\edge t_k = v$ in $G_\calC$. By construction of $G_\calC$, this path corresponds to constraints in set \ref{bool_form_2_reformulated} stipulating that $x_{t_0} = \ldots = x_{t_k},$ from which we get $x_u = x_v$ as $x$ satisfies constraint set \ref{bool_form_2_reformulated}. Conversely, assume $x$ is an assignment and that for each connected component $S$ of $G_\calC$ it holds that $x_u = x_v$ for all $u, v \in S$. Consider a constraint $x_u = x_v$ in set \ref{bool_form_2_reformulated}. By construction of $G_\calC$, the edge $u \edge v$ exists in $G_\calC$, so $u$ and $v$ are in the same connected component of $G_\calC,$ from which our assumption gives that $x_u = x_v$, so the constraint in \ref{bool_form_2_reformulated} is satisfied.
\end{proof}

\repeatlemma{lemma_complement_is_cc}
\begin{proof} By construction of $G_\calC$, an edge $a \edge b$ exists in $G_\calC$ if and only if the edge $\overline{a} \edge \overline{b}$ also exists in $G_\calC$. By induction, $a$ and $b$ are connected by a path in $G_\calC$ if and only if $\overline{a}$ and $\overline{b}$ are connected by a path in $G_\calC$. This means that vertex set $\overline{S}$ is connected in $G_\calC$. To show that $\overline{S}$ is a connected component, assume for a contradiction that $\overline{S}$ was contained in a larger connected component $\overline{S} \subsetneq S' \subseteq V(G_\calC)$. Note that this also means that $\overline{\overline{S}} = S \subsetneq \overline{S'}$. By the same argument as before applied to $S'$ instead of $S$, we get that $\overline{S'}$ is connected in $G_\calC$. However, $S \subsetneq \overline{S'}$, contradicting the fact that $S$ is a connected component.
\end{proof}

\subsection{Proof of Lemma \ref{lemma-middle}}

In this section, we prove Lemma \ref{lemma-middle}. Before commencing, we make an elementary observation:

\begin{proposition} \label{prop:unchanged_input} Appending to an approval matrix $\calP$ a copy of one of the rows of $\calP$ leaves the set of NB constraints $\calC_\calP$ unchanged. Naturally, the formula and colorful graphs also remain unchanged.
\end{proposition}

This fact will be useful in our proofs because it allows for assuming without loss of generality that any collection of NB constraints in $\calC_\calP$ is generated by pairwise disjoint unordered pairs of rows/candidates; e.g., if $(1, 2, 3)$ and $(4, 5, 6)$ are NB constraints witnessed by candidate pairs $\{a, b\}$ and $\{c, d\}$, then without loss of generality $\{a, b\} \cap \{c, d\} = \varnothing$. This is because had this not been the case, one could have added copies of the rows which are mentioned multiple times. 

To prove Lemma \ref{lemma-middle}, we will need two intermediate results that together will give the conclusion.


\begin{figure}[t]
    \begin{subfigure}{.3\textwidth}
        \centering
        \includegraphics[width=.8\linewidth]{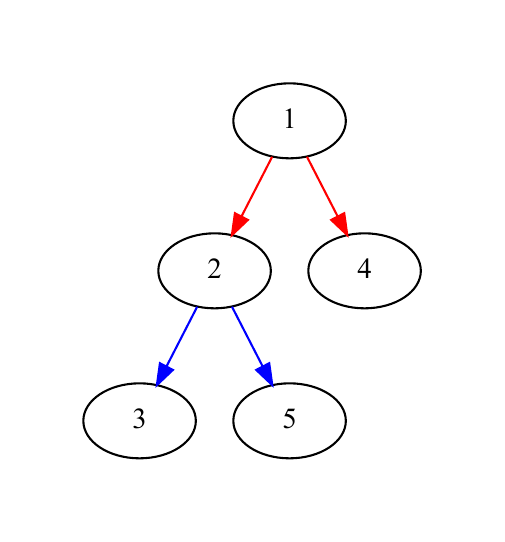}
        \caption{Assumed edge colors.}
        \label{fig:one-common-one-opposite-a}
    \end{subfigure}
    \begin{subfigure}{.3\textwidth}
        \centering
        \includegraphics[width=.8\linewidth]{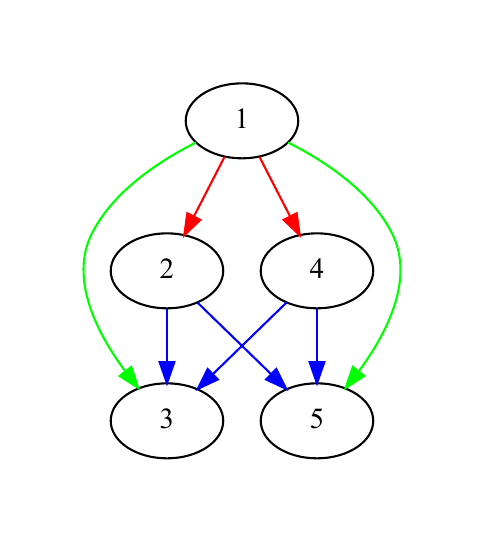}
        \caption{Implied edge colors.}
        \label{fig:one-common-one-opposite-b}
    \end{subfigure}
    \begin{subfigure}{.38\textwidth}
        \centering
        \includegraphics[width=.5\linewidth]{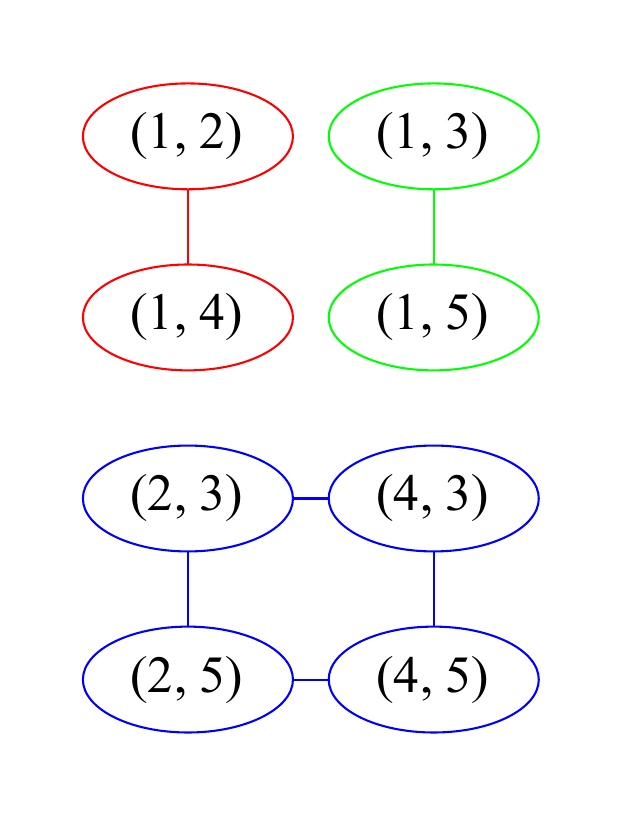}
        \vspace{0.25cm}
        \caption{Impl.~edges in formula graph.}
        \label{fig:one-common-one-opposite-c}
    \end{subfigure}
    \caption{Illustration of Lemma \ref{lemma-one-common-one-opposite}. Fig.~\ref{fig:one-common-one-opposite-a} depicts the assumptions:~edges $(1, 2) \edge (1, 4)$ and $(2, 3) \edge (2, 5)$ exist in the formula graph, with the first two colored red and the last two blue. Additionally, edge $(1, 3)$ is neither red nor blue. Fig.~\ref{fig:one-common-one-opposite-b} shows the edge colors implied by the setup, where green is some arbitrary color other than red and blue. Fig.~\ref{fig:one-common-one-opposite-c} shows the implied edges in the formula graph, whose existence lead to the color configuration in the second figure.}
    \label{fig:one-common-one-opposite}
\end{figure}
\begin{figure}[t]
    \centering
    \begin{subfigure}{.19\textwidth}
    \centering
    \begin{tabular}{*{5}{p{\matrixcolumnwidth}}}
        \hcr{0} & \hcr{1} & x & \hcr{1} & z \\
        \hcr{1} & \hcr{0} & y & \hcr{0} & t \\
        u & \hcb{0} & \hcb{1} & p & \hcb{1} \\
        v & \hcb{1} & \hcb{0} & q & \hcb{0} 
    \end{tabular}
    \caption{Notation.}
    \label{fig:lemma-one-common-one-opposite-matrices-a}
    \end{subfigure}
    \begin{subfigure}{.21\textwidth}
    \centering
    \begin{tabular}{*{5}{p{\matrixcolumnwidth}}}
        0 & 1 & \hc{0} & 1 & z \\
        1 & 0 & \hc{0} & 0 & t \\
        \hc{0} & 0 & 1 & p & 1 \\
        \hc{0} & 1 & 0 & q & 0 
    \end{tabular}
    \caption{$x, y, u, v$ are $0$.}
    \label{fig:lemma-one-common-one-opposite-matrices-b}
    \end{subfigure}
    \begin{subfigure}{.17\textwidth}
    \centering
    \begin{tabular}{*{5}{p{\matrixcolumnwidth}}}
        0 & 1 & 0 & 1 & z \\
        1 & 0 & 0 & 0 & \hc{0} \\
        0 & 0 & 1 & p & 1 \\
        0 & 1 & 0 & q & 0 
    \end{tabular}
    \caption{$t$ is $0$.}
    \label{fig:lemma-one-common-one-opposite-matrices-c}
    \end{subfigure}
    \begin{subfigure}{.19\textwidth}
    \centering
    \begin{tabular}{*{5}{p{\matrixcolumnwidth}}}
        0 & 1 & 0 & 1 & \hc{0} \\
        1 & 0 & 0 & 0 & 0 \\
        0 & 0 & 1 & p & 1 \\
        0 & 1 & 0 & q & 0 
    \end{tabular}
    \caption{$z$ is $0$.}
    \label{fig:lemma-one-common-one-opposite-matrices-d}
    \end{subfigure}
    \begin{subfigure}{.19\textwidth}
    \centering
    \begin{tabular}{*{5}{p{\matrixcolumnwidth}}}
        0 & 1 & 0 & 1 & 0 \\
        1 & 0 & 0 & 0 & 0 \\
        0 & 0 & 1 & \hc{0} & 1 \\
        0 & 1 & 0 & q & 0 
    \end{tabular}
    \caption{$p$ is $0$.}
    \label{fig:lemma-one-common-one-opposite-matrices-e}
    \end{subfigure}
    \caption{Steps in the proof of Lemma \ref{lemma-one-common-one-opposite}. The final matrix reached implies the conclusions regardless of the value of $q \in \{0, 1\}$.}
    \label{fig:lemma-one-common-one-opposite-matrices}
\end{figure}

\begin{lemma} \label{lemma-one-common-one-opposite} Consider five voters, say $1, 2, 3, 4$ and $5,$ such that the formula graph has edges $(1, 2) \edge (1, 4)$ and $(2, 3) \edge (2, 5)$, and that, moreover, $(1, 2)$ and $(2, 3)$ have different colors in the colorful graph. If edge $(1, 3)$ is either not present in the colorful graph, or it has a different color than $(1, 2)$ and $(2, 3),$ then the formula graphs is guaranteed to also contain edges $(1, 3) \edge (1, 5),$ $(4, 3) \edge (4, 5),$ $(2, 3) \edge (4, 3)$ and $(2, 5) \edge (4, 5)$. The scenario is illustrated in Fig.~\ref{fig:one-common-one-opposite}.
\end{lemma}

\begin{proof} Matrix $\calP$ induces NB constraints $(2, 1, 4)$ and $(3, 2, 5)$. By Proposition \ref{prop:unchanged_input} and the discussion afterward, we can assume that matrix $\calP$ has the form in Fig.~\ref{fig:lemma-one-common-one-opposite-matrices-a}. First, we prove that $x = y.$ Assume the contrary, and make a case distinction:

If $xy = 01,$ then $\calP$ induces the NB constraint $(1, 2, 3),$ so edge $(1, 2) \edge (3, 2)$ is in the formula graph, so $(1, 2)$ and $(3, 2)$ are in the same component, meaning that they bear the same color, contradicting the hypothesis.

If $xy = 10,$ then $\calP$ induces the NB constraint $(2, 1, 3),$ so edge $(2, 1) \edge (3, 1)$ is in the formula graph, so the component of edge $(3, 1)$ has size at least two, meaning that it occurs in the colored graph, contradicting the hypothesis.

Subsequently, we prove that $u = v.$ Assume the contrary. If $uv = 10,$ then $\calP$ induces the NB constraint $(1, 2, 3),$ and a contradiction follows similarly to the first case above. If $uv = 01,$ then $\calP$ induces the NB constraint $(2, 3, 1),$ from which the edge $(2, 3) \edge (1, 3)$ is in the formula graph, so the component of edge $(3, 1)$ has size at least two, meaning that it occurs in the colored graph, contradicting the hypothesis.

Now, let us prove that $x = u.$ Assume the contrary. If $xu = 01,$ then $\calP$ induces the NB constraint $(1, 2, 3),$ and a contradiction follows as before. If $xu = 10,$ then $\calP$ again induces the NB constraint $(1, 2, 3),$ again a contradiction.

Since $\calP$ and $\overline{\calP}$, which is the profile obtained from $\calP$ by exchanging zeros and ones, induce the same graphs, note that the cases with $x = 1$ and $x = 0$ are isomorphic with respect to complementing the matrix and permuting the row pairs $(1, 2)$ and $(3, 4)$. Therefore, it is enough to consider the case $x = 0,$ so our matrix has the form in Fig.~\ref{fig:lemma-one-common-one-opposite-matrices-b}.

Now, note that if $t = 1,$ then $\calP$ induces the NB constraint $(1, 2, 5)$, from which $(1, 2)$ and $(5, 2)$ have the same color in the colored graph, contradicting our hypothesis, so $t = 0$. Therefore, matrix $\calP$ has the form in Fig.~\ref{fig:lemma-one-common-one-opposite-matrices-c}. Notice that this induces the NB constraint $(3, 1, 5)$.

Moreover, observe that if $z = 1,$ then $\calP$ induces the NB constraint $(2, 1, 5),$ which, together with the presence of constraint $(3, 1, 5)$ means that $(1, 2)$ and $(1, 3)$ are colored with the same color, contradicting the hypothesis, so $z = 0$ holds. The matrix right now looks like in Fig.~\ref{fig:lemma-one-common-one-opposite-matrices-d}.

Finally, observe that if $p = 1,$ then $\calP$ induces the NB constraint $(4, 1, 5),$ which, together with the presence of constraints $(3, 1, 5)$ and $(2, 1, 4),$ means that $(1, 2)$ and $(1, 3)$ are colored with the same color, contradicting the hypothesis, so $p = 0$ holds. The matrix now looks like in Fig.~\ref{fig:lemma-one-common-one-opposite-matrices-e}. Looking only at the first three rows, we already get the conclusion. It can also be checked that both options for $q$ lead to no contradictions, but this is not necessary.
\end{proof}

\begin{figure}[t]
    \centering
\begin{subfigure}{.3\textwidth}
    \centering
    \includegraphics[width=.8\linewidth]{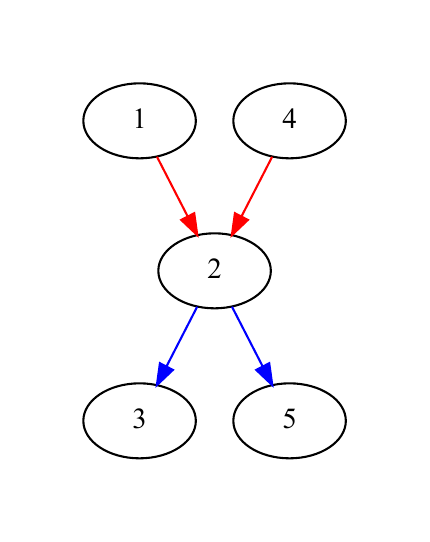}
    \caption{Assumed edge colors.}
    \label{fig:lemma-two-common-a}
\end{subfigure}
\begin{subfigure}{.3\textwidth}
    \centering
    \includegraphics[width=.92\linewidth]{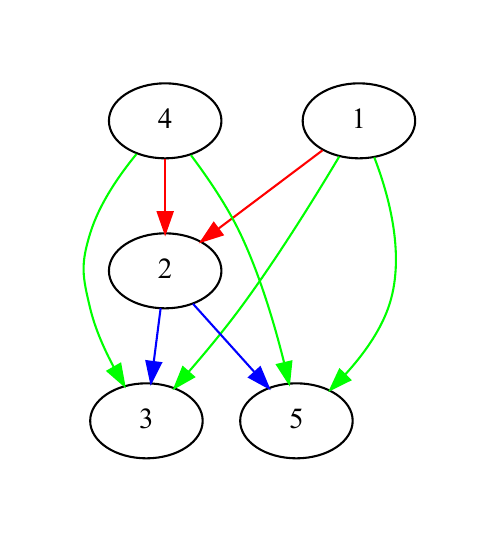}
    \caption{Implied edge colors.}
    \label{fig:lemma-two-common-b}
\end{subfigure}
\begin{subfigure}{.38\textwidth}
    \centering
    \includegraphics[width=.5\linewidth]{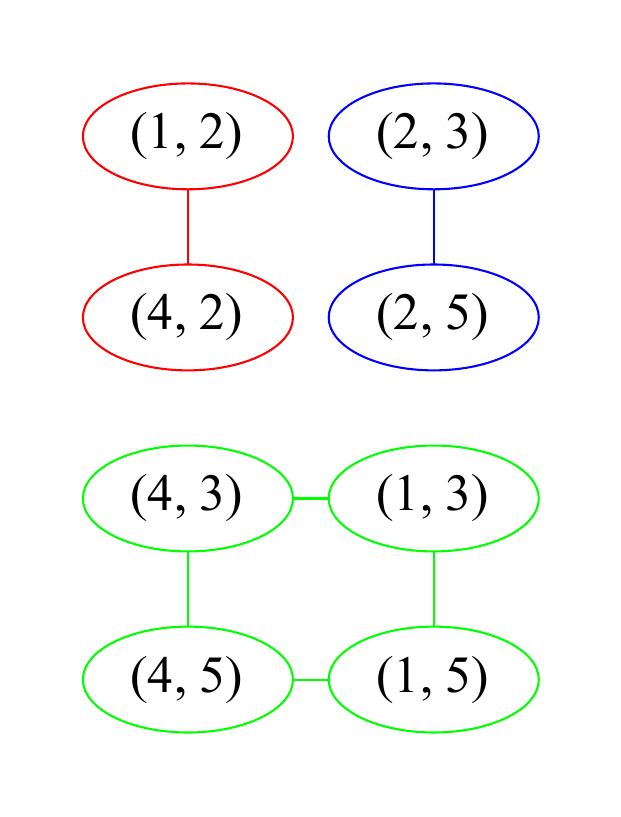}
    \vspace{0.35cm}
    \caption{Impl.~edges in formula graph.}
    \label{fig:lemma-two-common-c}
\end{subfigure}
\caption{Illustration of Lemma \ref{lemma-two-common}. Fig.~\ref{fig:lemma-two-common-a} depicts the assumptions:~edges $(1, 2) \edge (4, 2)$ and $(2, 3) \edge (2, 5)$ exist in the formula graph, with the first two colored red and the last two blue. Additionally, edge $(1, 3)$ is neither red nor blue. Fig.~\ref{fig:lemma-two-common-b} shows the edge colors implied by the setup, where green is some arbitrary color other than red and blue. Fig.~\ref{fig:lemma-two-common-c} shows the implied edges in the formula graph, whose existence lead to the color configuration in the second figure.}
\label{fig:lemma-two-common}
\end{figure}

Recall that a critical point to note is that Lemma \ref{lemma-one-common-one-opposite} continues to hold when some of the five voters happen to be the same, as can be seen  by creating distinct copies of such voters and then invoking the lemma. Similar considerations hold for all our results of this sort.

\begin{figure}[t]
    \centering
    \begin{subfigure}{.22\textwidth}
    \centering
    \begin{tabular}{*{5}{p{\matrixcolumnwidth}}}
        \hcr{0} & \hcr{1} & x & \hcr{0} & z \\
        \hcr{1} & \hcr{0} & y & \hcr{1} & t \\
        u & \hcb{0} & \hcb{1} & p & \hcb{1} \\
        v & \hcb{1} & \hcb{0} & q & \hcb{0} 
    \end{tabular}
    \caption{Notation.}
    \label{fig:lemma-two-common-matrices-a}
    \end{subfigure}
    \begin{subfigure}{.3\textwidth}
    \centering
    \begin{tabular}{*{5}{p{\matrixcolumnwidth}}}
        0 & 1 & \hc{0} & 0 & \hc{0} \\
        1 & 0 & \hc{0} & 1 & \hc{0} \\
        \hc{0} & 0 & 1 & p & 1 \\
        \hc{0} & 1 & 0 & q & 0 
    \end{tabular}
    \caption{$x, y, u, v, z, t$ are $0$.}
    \label{fig:lemma-two-common-matrices-b}
    \end{subfigure}
    \begin{subfigure}{.22\textwidth}
    \centering
    \begin{tabular}{*{5}{p{\matrixcolumnwidth}}}
        \hcb{0} & \hcb{1} & v & \hcb{0} & q \\
        \hcb{1} & \hcb{0} & u & \hcb{1} & p \\
        y & \hcr{0} & \hcr{1} & t & \hcr{1} \\
        x & \hcr{1} & \hcr{0} & z & \hcr{0} 
    \end{tabular}
    \caption{Self-symmetry.}
    \label{fig:lemma-two-common-matrices-c}
    \end{subfigure}
    \begin{subfigure}{.22\textwidth}
    \centering
    \begin{tabular}{*{5}{p{\matrixcolumnwidth}}}
        0 & 1 & 0 & 0 & 0 \\
        1 & 0 & 0 & 1 & 0 \\
        0 & 0 & 1 & \hc{0} & 1 \\
        0 & 1 & 0 & \hc{0} & 0 
    \end{tabular}
    \caption{$p, q$ are $0$.}
    \label{fig:lemma-two-common-matrices-d}
    \end{subfigure}
    
    \caption{Steps in the proof of Lemma \ref{lemma-two-common}.}
    \label{fig:lemma-two-common-matrices}
\end{figure}

\begin{lemma}\label{lemma-two-common} Consider five voters, say $1, 2, 3, 4$ and $5,$ such that the formula graph has edges $(1, 2) \edge (4, 2)$ and $(2, 3) \edge (2, 5)$, and that, moreover, $(1, 2)$ and $(2, 3)$ have different colors in the colorful graph. If edge $(1, 3)$ is either not present in the colorful graph, or it has a different color than $(1, 2)$ and $(2, 3),$ then the formula graphs is guaranteed to also contain edges $(4, 3) \edge (1, 3),$ $(4, 3) \edge (4, 5),$ $(1, 3) \edge (1, 5)$ and $(4, 5) \edge (1, 5)$. The scenario is illustrated in Fig.~\ref{fig:lemma-two-common}.
\end{lemma}

\begin{proof} Matrix $\calP$ induces NB constraints $(1, 2, 4)$ and $(3, 2, 5)$. By Proposition \ref{prop:unchanged_input} and the discussion afterwards, we can assume that matrix $\calP$ has the form in Fig.~\ref{fig:lemma-two-common-matrices-a}. Notice that matrices in Fig.~\ref{fig:lemma-two-common-matrices-a} and Fig.~\ref{fig:lemma-one-common-one-opposite-matrices-a} differ only in the two red entries in column 4 being swapped, so the proof of Lemma \ref{lemma-one-common-one-opposite} works unmodified until the first time column 4 is mentioned, so we get that it is enough to consider the case $x = y = u = v = z = t = 0$, depicted in Fig.~\ref{fig:lemma-two-common-matrices-b}. Moreover, notice that, in contrast to the previous lemma, the scenario in Fig.~\ref{fig:lemma-two-common-a} is symmetric with respect to interchanging the red/blue components. In particular, if we interchange the roles of candidate pairs $(1, 3)$ and $(4, 5)$ and additionally permute row pairs $(1, 4)$ and $(2, 3)$ in Fig.~\ref{fig:lemma-two-common-matrices-a} we get the matrix in Fig.~\ref{fig:lemma-two-common-matrices-c}, which is identical up to relabelling of the variables $x, y,$ etc. Consequently, reapplying our argument so far to the matrix in Fig.~\ref{fig:lemma-two-common-matrices-c} gives us that $v = u = y = x = q = p = 0.$ Altogether, we now know that $x = y = u = v = z = t = p = q = 0$, fact depicted in Fig.~\ref{fig:lemma-two-common-matrices-d}. Once again, looking only at the first three rows, we get our conclusion.
\end{proof}

We can now prove Lemma \ref{lemma-middle} by combining Lemmas \ref{lemma-one-common-one-opposite} and \ref{lemma-two-common}, as follows:

\repeatlemma{lemma-middle}

\begin{proof} Since edge $(2, 3)$ is in the colorful graph, an edge $(2, 3) \edge (x, y)$ has to exist in the formula graph. By construction of the formula graph, either $x = 2$ or $y = 3$. If $x = 2,$ then Lemma \ref{lemma-two-common} gives the conclusion. If $y = 3,$ then Lemma \ref{lemma-one-common-one-opposite} gives the conclusion.
\end{proof}

\subsection{Proof of Theorem \ref{lemma-no-edge-cliques}}

In this section, we prove Theorem \ref{lemma-no-edge-cliques}. In order to do so, we will need yet another intermediate result similar in spirit to Lemmas \ref{lemma-one-common-one-opposite} and \ref{lemma-two-common}, as follows:

\begin{figure}[t]
    \centering
\begin{subfigure}{.3\textwidth}
    \centering
    \includegraphics[width=.9\linewidth]{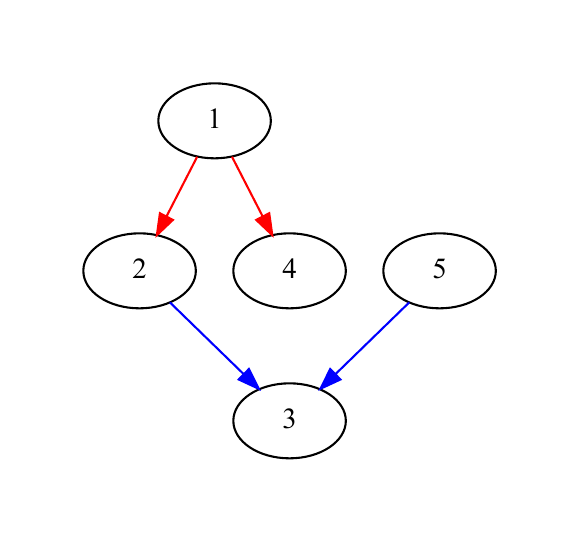}
    \caption{Assumed edge colors.}
    \label{fig:lemma-two-opposite-a}
\end{subfigure}
\begin{subfigure}{.3\textwidth}
    \centering
    \includegraphics[width=.9\linewidth]{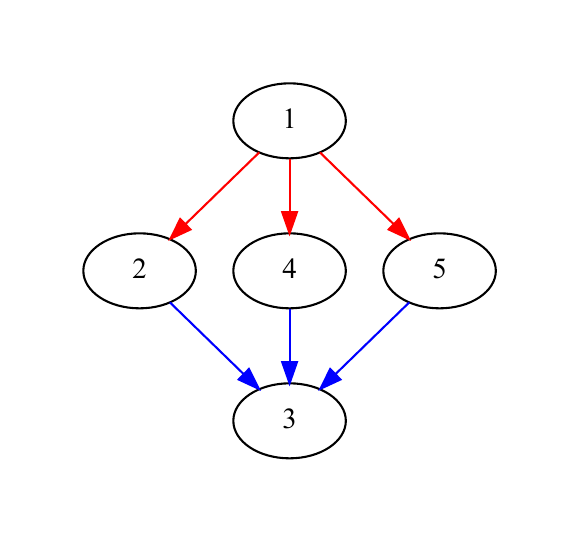}
    \caption{Implied edge colors.}
    \label{fig:lemma-two-opposite-b}
\end{subfigure}
\begin{subfigure}{.38\textwidth}
    \centering
    \includegraphics[width=.5\linewidth]{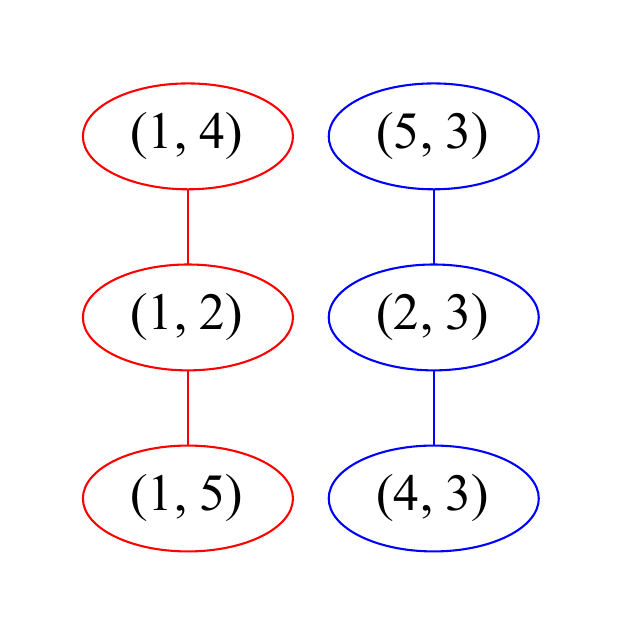}
    \vspace{0.52cm}
    \caption{Impl.~edges in formula graph.}
    \label{fig:lemma-two-opposite-c}
\end{subfigure}
\caption{Illustration of Lemma \ref{lemma-two-opposite}. Fig.~\ref{fig:lemma-two-opposite-a} depicts the assumptions:~edges $(1, 2) \edge (1, 4)$ and $(2, 3) \edge (5, 3)$ exist in the formula graph, with the first two colored red and the last two blue. Additionally, edge $(1, 3)$ is neither red nor blue. Fig.~\ref{fig:lemma-two-opposite-b} shows the edge colors implied by the setup. Fig.~\ref{fig:lemma-two-opposite-c} shows the implied edges in the formula graph, whose existence lead to the color configuration in the second figure.}
\label{fig:lemma-two-opposite}
\end{figure}
\begin{figure}[t]
    \centering
    \begin{subfigure}{.22\textwidth}
    \centering
    \begin{tabular}{*{5}{p{\matrixcolumnwidth}}}
        \hcr{0} & \hcr{1} & x & \hcr{1} & z \\
        \hcr{1} & \hcr{0} & y & \hcr{0} & t \\
        u & \hcb{0} & \hcb{1} & p & \hcb{0} \\
        v & \hcb{1} & \hcb{0} & q & \hcb{1} 
    \end{tabular}
    \caption{Notation.}
    \label{fig:lemma-two-opposite-matrices-a}
    \end{subfigure}
    \begin{subfigure}{.22\textwidth}
    \centering
    \begin{tabular}{*{5}{p{\matrixcolumnwidth}}}
        0 & 1 & \hc{0} & 1 & z \\
        1 & 0 & \hc{0} & 0 & t \\
        \hc{0} & 0 & 1 & p & 0 \\
        \hc{0} & 1 & 0 & q & 1 
    \end{tabular}
    \caption{$x, y, u, v$ are $0$.}
    \label{fig:lemma-two-opposite-matrices-b}
    \end{subfigure}    
    \begin{subfigure}{.22\textwidth}
    \centering
    \begin{tabular}{*{5}{p{\matrixcolumnwidth}}}
        0 & 1 & 0 & 1 & z \\
        1 & 0 & 0 & 0 & \hc{0} \\
        0 & 0 & 1 & p & 0 \\
        0 & 1 & 0 & q & 1 
    \end{tabular}
    \caption{$t$ is $0$.}
    \label{fig:lemma-two-opposite-matrices-c}
    \end{subfigure}
    \begin{subfigure}{.22\textwidth}
    \centering
    \begin{tabular}{*{5}{p{\matrixcolumnwidth}}}
        0 & 1 & 0 & 1 & z \\
        1 & 0 & 0 & 0 & 0\\
        0 & 0 & 1 & \hc{0} & 0 \\
        0 & 1 & 0 & q & 1 
    \end{tabular}
    \caption{$p$ is $0$.}
    \label{fig:lemma-two-opposite-matrices-d}
    \end{subfigure}
    \caption{Steps in the proof of Lemma \ref{lemma-two-opposite}. The final matrix reached implies the conclusions regardless of the values of $(q, z) \in \{0, 1\}^2$.}
    \label{fig:lemma-two-opposite-matrices}
\end{figure}
\begin{lemma}\label{lemma-two-opposite} Consider five voters, say $1, 2, 3, 4$ and $5,$ such that the formula graph has edges $(1, 2) \edge (1, 4)$ and $(2, 3) \edge (5, 3)$, and that, moreover, $(1, 2)$ and $(2, 3)$ have different colors in the colorful graph. If edge $(1, 3)$ is either not present in the colorful graph, or it has a different color than $(1, 2)$ and $(2, 3),$ then the formula graphs is guaranteed to also contain edges $(1, 2) \edge (1, 5)$ and $(2, 3) \edge (4, 3)$. The scenario is illustrated in Fig.~\ref{fig:lemma-two-opposite}.
\end{lemma}
\begin{proof}
Matrix $\calP$ induces NB constraints $(2, 1, 4)$ and $(2, 3, 5)$. By Proposition \ref{prop:unchanged_input} and the discussion afterwards, we can assume that matrix $\calP$ has the form in Fig.~\ref{fig:lemma-two-opposite-matrices-a}. Notice that matrices in Fig.~\ref{fig:lemma-two-opposite-matrices-a} and Fig.~\ref{fig:lemma-one-common-one-opposite-matrices-a} differ only in the two blue entries in column 5 being swapped, so the proof of Lemma \ref{lemma-one-common-one-opposite} works unmodified until the first time column 5 is mentioned, so we get that it is enough to consider the case $x = y = u = v = 0$, depicted in Fig.~\ref{fig:lemma-two-opposite-matrices-b}.

Now, notice that if $t = 1,$ then $\calP$ induces the NB constraint $(1, 3, 5),$ from which $(1, 3)$ and $(5, 3)$ have the same color in the colorful graph, which happens to also be the color of $(2, 3),$ contradicting our hypothesis, so $t = 0$ holds, fact depicted in Fig.~\ref{fig:lemma-two-opposite-matrices-c}.

Similarly, notice that if $p = 1,$ then $\calP$ induces the NB constraint $(3, 1, 4),$ from which $(1, 3)$ and $(1, 4)$ have the same color in the colorful graph, which happens to also be the color of $(1, 2),$ contradicting our hypothesis, so $p = 0$ holds, fact depicted in Fig.~\ref{fig:lemma-two-opposite-matrices-d}. As a remark, observe that we could have also argued that Fig.~\ref{fig:lemma-two-common-a} is symmetric with respect to interchanging red/blue components, and then use a symmetry argument like in the proof of Lemma \ref{lemma-two-common} to deduce that $p = 0,$ but this time not much work is saved in contrast to a direct argument.

Finally, irrespective of the values of $q$ and $t$, the current matrix $\calP$ already induces NB constraints $(2, 1, 5)$ and $(2, 3, 4)$, implying our conclusion. It can also be checked that all four options for $(z, q) \in \{0, 1\}^2$ lead to no contradictions, but this is not necessary.
\end{proof}

In order to streamline the proof of Theorem \ref{lemma-no-edge-cliques}, we summarize the essence of Lemmas \ref{lemma-one-common-one-opposite} and \ref{lemma-two-opposite} into the following lemma, similar in spirit to Lemma \ref{lemma-middle}:

\begin{figure}[t]
    \centering
    \begin{subfigure}{.25\textwidth}
    \centering
    \includegraphics[width=.95\linewidth]{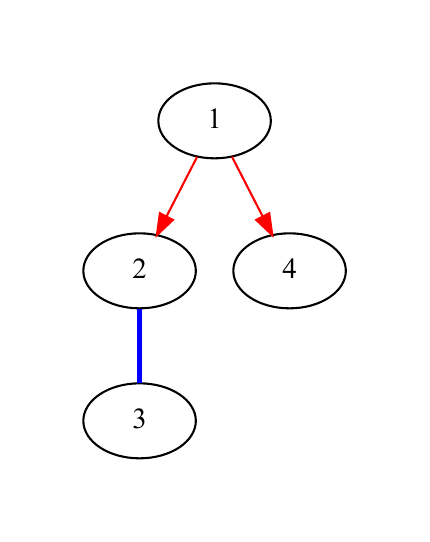}
    \caption{Assumed colors.}
    \label{fig:lemma-ends-a}
    \end{subfigure}
    \begin{subfigure}{.25\textwidth}
    \centering
    \includegraphics[width=.95\linewidth]{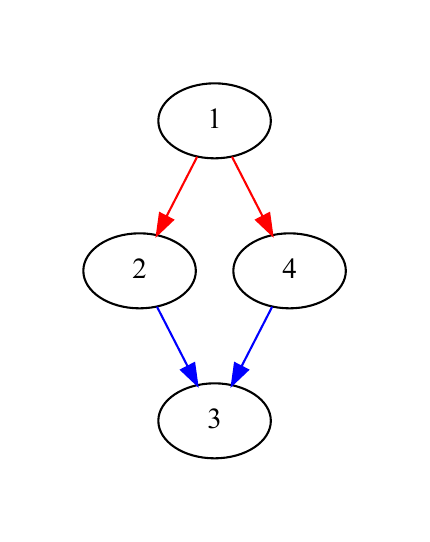}
    \caption{Implied colors.}
    \label{fig:lemma-ends-b}
    \end{subfigure}
    \caption{Illustration of Lemma \ref{lemma-ends}. Fig.~\ref{fig:lemma-ends-a} depicts the assumptions:~edge $(1, 2) \edge (1, 4)$ exists in the formula graph (colored red) and edge $(2, 3)$ exists in the colorful graph and bears a different color (colored blue). Additionally, edge $(1, 3)$ is neither red nor blue. Fig.~\ref{fig:lemma-ends-b} shows the colors implied by the setup.
    }
    \label{fig:lemma-ends}
\end{figure}
\begin{lemma} \label{lemma-ends} Consider four voters, say $1, 2, 3, 4,$ such that the formula graph contains the edge $(1, 2) \edge (1, 4)$ and that edge $(2, 3)$ exists in the colorful graph and has a different color than $(1, 2).$ If edge $(1, 3)$ is either not present in the colorful graph, or it has a different color than $(1, 2)$ and $(2, 3)$, then edge $(2, 3) \edge (4, 3)$ is guaranteed to be in the formula graph. The scenario is illustrated in Fig.~\ref{fig:lemma-ends-a} and Fig.~\ref{fig:lemma-ends-b}.
\end{lemma}
\begin{proof}
Since edge $(2, 3)$ is in the colorful graph, an edge $(2, 3) \edge (x, y)$ has to exist in the formula graph. By construction of the formula graph, either $x = 2$ or $y = 3$. If $x = 2,$ then Lemma \ref{lemma-one-common-one-opposite} gives the conclusion. If $y = 3,$ then Lemma \ref{lemma-two-opposite} gives the conclusion.
\end{proof}

We are now ready to prove Theorem \ref{lemma-no-edge-cliques}, as follows:

\repeattheorem{lemma-no-edge-cliques}
\begin{proof} Consider an inclusion-maximal set $S$ such that $2 \in S,$ $1, 3 \notin S$ and the subgraphs of the formula graph induced by $\{1\} \times S$ and $S \times \{3\}$ are connected. Note that the connectivity assumption implies that in the colorful graph edges in $\{1\} \times S$ have the color of $(1, 2)$ and edges in $S \times \{3\}$ have the color of $(2, 3)$. Denote by $C_{12}$ and $C_{23}$ the connected components of $(1, 2)$ and $(2, 3)$ in the formula graph. We will now show that $C_{12} = \{1\} \times S$ and $C_{23} = S \times \{3\}.$ Assume for a contradiction that this was not the case, then, by symmetry, without loss of generality, $C_{12} \neq \{1\} \times S$. Together with the above, this means that $\{1\} \times S \subsetneq C_{12}.$ Because $C_{12}$ induces a connected subgraph in the formula graph, and $\varnothing \neq \{1\} \times S \subsetneq C_{12},$ it follows that there is an edge in the formula graph crossing the cut $(\{1\} \times S, C_{12} \setminus (\{1\} \times S))$.
This edge can be of one of two forms:~either $(1, s) \edge (x, s)$ or $(1, s) \edge (1, x),$ where $s \in S$ and $x \notin \{1, s\}$. We tackle the two cases separately in the paragraphs below. 

If the edge crossing the cut is of the form $(1, s) \edge (x, s),$ then instantiate Lemma \ref{lemma-middle} with $1 \mapsto 1, 2 \mapsto s, 3 \mapsto 3$ and $4 \mapsto x$. This is sound because the edge $(1, s) \edge (x, s)$ exists in the formula graph and $(1, s),$ $(s, 3)$ have different colors in the colorful graph, together with the fact that $(1, 3)$ does not appear in the colorful graph. As a result, we get that edge $(1, 3) \edge (x, 3)$ is in the formula graph, contradicting the nonexistence of edge $(1, 3)$ in the colorful graph.

On the other hand, if the edge crossing the cut is of the form $(1, s) \edge (1, x),$ then this time we instantiate Lemma \ref{lemma-ends}, with values $1 \mapsto 1, 2 \mapsto s, 3 \mapsto 3$ and $4 \mapsto x$. This is sound because edge $(1, s) \edge (1, x)$ is in the formula graph and $(1, s), (s, 3)$ have different colors in the colorful graph, together with the fact that $(1, 3)$ does not appear in the colorful graph. As a result, we get that edge $(s, 3) \edge (x, 3)$ appears in the formula graph. As a result, the subgraphs induced by $\{1\} \times (S \cup \{x\})$ and $(S \cup \{x\}) \times \{3\}$ in the formula graph are connected. If we could also show that $x \notin (S \cup \{1, 3\})$ then we would get that our $S$ was not maximal with respect to inclusion, a contradiction. We now show this last fact to conclude the proof:~$x = 1$ was forbidden above; $x = 3$ would make our cut edge be $(1, s) \edge (1, 3)$, implying that $(1, 3)$ appeared in the colorful graph, a contradiction;~finally $x \in S$ would mean that $(1, s) \edge (1, x)$ is not a cut edge.
\end{proof}

\subsection{Proof of Lemma \ref{lemma-stronger-middle}}

In this section, we prove Lemma \ref{lemma-stronger-middle}, which we repeat below for convenience:

\repeatlemma{lemma-stronger-middle}

The proof closely follows that of Lemma \ref{lemma-middle}, except for requiring stronger versions of Lemmas \ref{lemma-one-common-one-opposite} and \ref{lemma-two-common}, which we state and prove next. The proofs remain similar in spirit, by identifying small matrices with blanks to be filled in, like Fig.~\ref{fig:lemma-one-common-one-opposite-matrices-a}. Unfortunately, the weaker assumptions make it rather challenging to proceed as before in a principled manner. Instead, we employ a computer program to try out all the possibilities.

\begin{lemma}[Strengthened Lemma \ref{lemma-one-common-one-opposite}]\label{stronger-lemma-one-common-one-opposite} Consider five voters, say $1, 2, 3, 4$ and $5,$ such that the formula graph has edges $(1, 2) \edge (1, 4)$ and $(2, 3) \edge (2, 5)$, and that, moreover, the formula graph does not contain edges $(1, 2) \edge (3, 2)$ or $(1, 2) \edge (5, 2)$. Then, the formula graphs is guaranteed to contain the edge $(1, 3) \edge (1, 5).$
\end{lemma}
\begin{proof} Similarly to the proof of Lemma \ref{lemma-one-common-one-opposite}, we can assume that matrix $\calP$ has the form in Fig.~\ref{fig:lemma-one-common-one-opposite-matrices-a}. However, this time our premises are weaker, so we did not manage to write a direct ``linear'' argument. However, we tried out all $2^8 = 256$ options for the tuple $(x, y, z, t, u, v, p, q)$ with a computer and $68$ of them did not contradict our hypothesis, and for all of them the conclusion held. The relevant code is included with the paper, and instructions for running it can be found in Appendix \ref{app:code}.
\end{proof}

\begin{lemma}[Strengthened Lemma \ref{lemma-two-common}]\label{stronger-lemma-two-common} Consider five voters, say $1, 2, 3, 4$ and $5,$ such that the formula graph has edges $(1, 2) \edge (4, 2)$ and $(2, 3) \edge (2, 5)$, and that, moreover, the formula graph does not contain edges $(1, 2) \edge (3, 2)$ or $(1, 2) \edge (5, 2).$ Then, the formula graphs is guaranteed to contain the edge $(1, 3) \edge (1, 5).$
\end{lemma}
\begin{proof} Similarly to the proof of Lemma \ref{lemma-two-common}, we can assume that matrix $\calP$ has the form in Fig.~\ref{fig:lemma-two-common-matrices-a}. However, this time our premises are weaker, so we did not manage to write a direct ``linear'' argument. However, we tried out all $2^8 = 256$ options for the tuple $(x, y, z, t, u, v, p, q)$ with a computer and $59$ of them did not contradict our hypothesis, and for all of them the conclusion held. The relevant code is included with the paper, and instructions for running it can be found in Appendix \ref{app:code}. 
\end{proof}

\subsection{Code for Proofs of Lemmas \ref{stronger-lemma-one-common-one-opposite} and \ref{stronger-lemma-two-common}} \label{app:code}

Together with the paper, we include a \textsc{C++} source file called \texttt{exhaust\_proof.cpp}. This file can be found at \url{https://github.com/Andrei1998/rec-sc-from-approval/}.\\

\noindent To compile the code, use the following command:

\begin{center}
\texttt{g++ -std=c++11 -O2 exhaust\_proof.cpp -o exhaust\_proof}
\end{center}

\noindent To run the check required in the proof of Lemma \ref{stronger-lemma-one-common-one-opposite}, invoke the compiled program as:

\begin{center}
\texttt{echo "1" | ./exhaust\_proof}
\end{center}

\noindent To run the check required in the proof of Lemma \ref{stronger-lemma-two-common}, invoke the compiled program as:

\begin{center}
\texttt{echo "2" | ./exhaust\_proof}
\end{center}



\end{document}